\documentclass[letterpaper,11pt]{article}

\usepackage[margin=1in]{geometry}

\usepackage{amsthm}
\usepackage{amsfonts}
\usepackage{amssymb}
\usepackage{amsmath}
\usepackage{booktabs}
\usepackage[hidelinks]{hyperref}
\usepackage{makeidx}
\usepackage{tabularx, booktabs}
\usepackage{tabu}
\usepackage{colortbl}
\usepackage{enumitem}
\usepackage{multirow}
\usepackage{graphicx}
\usepackage{booktabs}
\usepackage{listings}
\usepackage{marginnote}
\usepackage{comment}
\usepackage{thmtools}
\usepackage{thm-restate}
\usepackage{authblk}
\usepackage[noend,ruled, linesnumbered]{algorithm2e}
\usepackage[capitalize]{cleveref}
\usepackage[utf8]{inputenc}

\setitemize{itemsep=0pt,topsep=5pt,parsep=0pt,partopsep=0pt}
\setenumerate{itemsep=0pt,topsep=5pt,parsep=0pt,partopsep=0pt}

\usepackage{graphicx}
\graphicspath{{./figures/}}

\SetKwComment{tcr}{$\triangleright$\ }{}
\SetCommentSty{textsf}

\newtheorem{theorem}{Theorem}[section]
\newtheorem{lemma}[theorem]{Lemma}
\newtheorem{fact}[theorem]{Fact}
\newtheorem{observation}[theorem]{Observation}
\newtheorem{corollary}[theorem]{Corollary}
\newtheorem{definition}[theorem]{Definition}

\theoremstyle{definition}
\newtheorem{problem}{Problem}
\newtheorem{remark}[theorem]{Remark}

\newtheorem{example}[theorem]{Example}

\crefname{problem}{Problem}{Problems}
\crefname{miniclaim}{Claim}{Claims}

\newcommand{\FF}{{\cal F}}
\newcommand{\HH}{{\cal H}}

\newcommand{\down}[1]{\left\lfloor#1\right\rfloor}

\renewcommand{\O}{\mathcal{O}}
\newcommand{\ceil}[1]{\left\lceil{#1}\right\rceil}
\newcommand{\floor}[1]{\left\lfloor{#1}\right\rfloor}

\newcommand{\HAM}{\mathsf{Ham}}

\newcommand{\num}{t}

\newcommand{\DD}{\widetilde{d}}

\newcommand{\dd}{\mathinner{.\,.}}

\newcommand{\Z}{\mathbb{Z}}

\newcommand{\E}{\mathbb{E}}

\newcommand{\F}{\mathcal{F}}
\newcommand{\G}{\mathcal{G}}

\newcommand{\Oh}{O}
\newcommand{\Ohtilde}{\tilde{\Oh}}
\newcommand{\sub}{\subseteq}
\newcommand{\supp}{\mathrm{supp}}

\newcommand{\eps}{\varepsilon}
\newcommand{\teps}{\tilde \eps}

\newcommand{\tO}{\tilde{O}}
\newcommand{\OO}{\hat O}
\newcommand{\OCC}{\textit{occ}}
\newcommand{\tteps}{\hat{\eps}}

\def\polylog{\operatorname{polylog}}

\newcommand{\defproblem}[4]{
  \vspace{2mm}
  \noindent\fbox{
  \begin{minipage}{0.96\textwidth}
\begin{problem}\label{#1}
  \textsf{#2}

  \smallskip
  \noindent
  {\bf{Input:}} #3

  \smallskip
  \noindent
  {\bf{Output:}} #4
\end{problem}
  \end{minipage}
  }
  \vspace{2mm}
}

\usepackage{xcolor}

\definecolor{purple}{rgb}{0.58, 0.44, 0.86}

\makeatletter
\newcommand*{\centerfloat}{%
	\parindent \z@
	\leftskip \z@ \@plus 1fil \@minus \textwidth
	\rightskip\leftskip
	\parfillskip \z@skip}
\makeatother

\title{Approximating Text-to-Pattern Hamming Distances}

\author[1]{Timothy M. Chan}
\author[2]{Shay Golan}
\author[2]{Tomasz Kociumaka}
\author[2]{Tsvi Kopelowitz}
\author[2]{Ely Porat}
\affil[1]{Department of Computer Science, University of Illinois at Urbana-Champaign, IL, USA}
\affil[ ]{\texttt{tmc@illinois.edu}}
\affil[2]{Department of Computer Science, Bar-Ilan University, Ramat Gan, Israel}
\affil[ ]{\texttt{golansh1@cs.biu.ac.il, kociumaka@mimuw.edu.pl, kopelot@gmail.com, porately@cs.biu.ac.il}}

\date{\vspace{-1cm}}

\begin{document}

\maketitle

\begin{abstract}
We revisit a fundamental problem in string matching: given a pattern of length $m$ and a text of length $n$, both over an alphabet of size $\sigma$, compute the Hamming distance (i.e., the number of mismatches) between the pattern and the text at every location.
Several randomized $(1+\varepsilon)$-approximation algorithms have been proposed in the literature (e.g., by Karloff (Inf.\ Proc.\ Lett., 1993), Indyk (FOCS 1998), and Kopelowitz and Porat (SOSA 2018)), with running time of the
form $O(\varepsilon^{-O(1)}n\log n\log m)$, all using fast Fourier transform (FFT\@).  We describe a simple randomized $(1+\varepsilon)$-approximation algorithm that is faster and does not need FFT\@.  
Combining our approach with additional ideas leads to numerous new results (all Monte-Carlo randomized) in different settings:

\begin{enumerate}
\item We obtain
the first \emph{linear-time} approximation algorithm; the running time is $O(\varepsilon^{-2}n)$.  In fact, the time bound can be made slightly sublinear in $n$ if the alphabet size $\sigma$ is small (by using bit packing tricks).
\item We apply our approximation algorithms to obtain a faster \emph{exact} algorithm computing all Hamming distances up to a given threshold $k$; its running time is $O(n + \min(\frac{nk\sqrt{\log m}}{\sqrt{m}},\frac{nk^2}{m}))$, which improves previous results by logarithmic factors and is linear if $k\le \sqrt{m}$.
\item We alternatively obtain approximation algorithms with better \emph{$\varepsilon$-dependence}, by using rectangular matrix multiplication. 
In fact, the time bound is $O(n \polylog n)$
when the pattern is sufficiently long, i.e., $m\ge \varepsilon^{-c}$ for a specific constant $c$.  Previous algorithms with the best $\varepsilon$-dependence require $O(\varepsilon^{-1}n\polylog n)$ time.
\item
When $k$ is not too small,
we obtain a truly \emph{sublinear-time} algorithm to find all locations with Hamming distance approximately (up to a constant factor) less than $k$, in 
$O((n/k^{\Omega(1)}+\textit{occ})n^{o(1)})$ time, where $\textit{occ}$ is the output size. The algorithm leads to a \emph{property tester} for pattern matching,
with high probability returning true if an exact match exists and false if the Hamming distance is more than $\delta m$ at every location, running in
$O((\delta^{-1/3}n^{2/3} + \delta^{-1}\frac nm)\polylog n)$ time.
\item
We obtain a \emph{streaming} algorithm to report all locations
with Hamming distance approximately less than $k$, using
$O(\eps^{-2}\sqrt{k}\polylog n)$ space.  Previously,
streaming algorithms were known for the exact problem with
$O(k\polylog n)$ space (which is tight up to the $\polylog n$ factor) or for the approximate
problem with  $O(\eps^{-O(1)}\sqrt{m}\polylog n)$ space. 
For the special case of $k=m$, we improve the space usage to $O(\eps^{-1.5}\sqrt{m}\polylog n)$.

\end{enumerate}
\end{abstract}
\setcounter{page}{0}
\thispagestyle{empty}

\newpage

\section{Introduction}\label{sec:intro}

We study a fundamental problem in string matching: given a pattern of length $m$ and a text of length $n$ over an alphabet of size $\sigma$,
compute the Hamming distance (i.e., the number of mismatches)
between the pattern and the text at every location.
Of particular interest is the version with a fixed threshold, known as the
\emph{$k$-mismatch} problem: compute the Hamming distances only for locations with distances at most a given value~$k$.
This includes as a special case the decision problem of testing
whether the Hamming distance is at most $k$ at each location (in particular, deciding whether there exists a location with at most $k$ mismatches).

The problem has an extensive history, spanning over four decades; see \cref{tbl} for a summary.  For arbitrary~$\sigma$, the best time bound, $\tO(n + \frac{nk}{\sqrt{m}})$,\footnote{
Throughout the paper, $\tO$ hides polylogarithmic factors, and $\OO$ hides
$n^{o(1)}$ factors.  Additionally, $\tO_\eps$ and $\OO_\eps$ may hide $\eps^{-O(1)}$ factors.
}
 by Gawrychowski and Uznański (ICALP 2018)~\cite{GawUzn}, subsumes all previous bounds up to logarithmic factors;
their paper also provides conditional lower bounds suggesting that no substantially faster ``combinatorial'' algorithms are possible.

\begin{table}[h]
\renewcommand{\arraystretch}{1.25}
\centering
\begin{tabular}{ll}
\toprule
Fischer and Paterson~\cite{FisPat} & $O(\sigma n\log m)$\\
Abrahamson~\cite{DBLP:journals/siamcomp/Abrahamson87} & $O(n\sqrt{m\log m})$\\
Landau and Vishkin~\cite{LanVis,LanVis89} / 
Galil and Giancarlo~\cite{GalGia} & $O(nk)$\\
Sahinalp and Vishkin~\cite{SahVis} & $O(n + \frac{nk^{O(1)}}{m})$ \\
Cole and Hariharan~\cite{ColHar} & $O(n + \frac{nk^4}{m})$ \\
Amir, Lewenstein, and Porat~\cite{Ami} & $O(n\sqrt{k\log k})$ \\
Amir, Lewenstein, and Porat~\cite{Ami} & $O(n\log k + \frac{nk^3\log k}{m})$\\
Clifford, Fontaine, Porat,
Sach, and Starikovskaya~\cite{Cli}
& $O(n\polylog m+ \frac{nk^2\log k}{m})$\\
Gawrychowski and Uznański~\cite{GawUzn} &
$O(n\log^2m\log\sigma + \frac{nk\sqrt{\log n}}{\sqrt{m}})$\\\bottomrule
\end{tabular}
\caption{Time bounds of known exact algorithms for computing all distances at most $k$.}\label{tbl}
\end{table}

As a function of $n$, the time bound is $\tO(n^{3/2})$ in the worst case, when $m$ and $k$ are linear in~$n$.  To obtain faster algorithms, researchers have turned to the \emph{approximate} version of the problem: finding values that are within a $1+\eps$ factor of the true distances.

Several efficient randomized (Monte-Carlo) algorithms for approximating all Hamming distances have been proposed.  There are three main simple approaches:
\begin{itemize}
\item  Karloff (Inf.\ Proc.\ Lett., 1993)~\cite{Kar} obtained an $O(\eps^{-2} n\log n\log m)$-time algorithm, by randomly mapping the alphabet to $\{0,1\}$, thereby reducing the problem to
$O(\eps^{-2}\log n)$ instances with $\sigma=2$.  Each such instance can be solved in $O(n\log m)$ time by standard convolution, i.e., fast Fourier transform (FFT)\@.
Karloff's approach can be derandomized (in $O(\eps^{-2} n\log^3m)$ time, via $\eps$-biased sample spaces or error-correcting codes).
\item
Indyk (FOCS 1998)~\cite{Ind} solved the approximate decision problem
for a fixed threshold 
in $O(\eps^{-3} n\log n)$ time, by
using random sampling 
and performing $O(\eps^{-3}\log n)$ convolutions in $\mathbb{F}_2$, each
doable in $O(n)$ time by a bit-packed version of FFT\@.
The general problem can then be solved by examining logarithmically many thresholds, in $O(\eps^{-3} n\log n\log m)$ time.
\item
Kopelowitz and Porat (SOSA 2018)~\cite{KopPor} obtained an
$O(\eps^{-1}n\log n\log m)$-time algorithm, by randomly mapping
the alphabet to $[O(\eps^{-1})]$,\footnote{Throughout the paper, let $[x]=\{0,1,\ldots,x-1\}$ for a positive integer $x$.}
thereby reducing the problem to 
$O(\log n)$ instances with $\sigma=O(\eps^{-1})$.  Each such instance can be solved
by $O(\eps^{-1})$ convolutions.   This result is notable for its better
$\eps$-dependence; previously, Kopelowitz and Porat (FOCS 2015)~\cite{KopPorFOCS15} gave a more complicated
algorithm~\cite{KopPorFOCS15} with $O(\eps^{-1}n\log n\log m\log\sigma\log(1/\eps))$ randomized running time (which also uses FFT).
\end{itemize}

\medskip
All three algorithms require $O(n\log^2n)$ time as a function of $n$, and they all use FFT\@.  Two natural questions arise: (i) can the $n\log^2n$ barrier be broken? (ii) is FFT necessary for obtaining nearly linear time algorithms?

\subsection{A New Simple Approximation Algorithm}
In \cref{sec:generic,sec:offline}, we present a  randomized approximation algorithm  which costs $O(\eps^{-2.5} n\log^{1.5}n)$ time and does not use FFT, thereby answering both questions.  
As in previous randomized algorithms, the algorithm is Monte-Carlo and its results are correct with high probability, i.e., the error probability is $O(n^{-c})$ for an arbitrarily large constant~$c$.

Our approach is based on random sampling (like Indyk's~\cite{Ind}): the Hamming distance is estimated by checking mismatches at a random subset of positions. 
In order to avoid FFT, our algorithm uses a random subset with more structure: the algorithm picks a random prime $p$ (of an appropriately chosen size) and a random offset $b$, and considers a subset of positions
$\{b,b+p,b+2p,\ldots\}$.
The structured nature of the subset enables more efficient computation.  
It turns out that even better efficiency is achieved
by using multiple (but still relatively few) offsets.

When approximating the Hamming distance of the pattern at subsequent text locations, the set of sampled positions in the text changes, and so a straightforward implementation seems too costly. 
To overcome this challenge, a key idea is to shift the sample a few times in the pattern and a few times in the text (namely,  for a tradeoff parameter~$z$, our algorithm considers $z$ shifts in the pattern and $p/z$ shifts in the text).

While these simple ideas individually may have appeared before in one form or another in the literature, we demonstrate that they are quite powerful when put together in the right way, and with the right choice of parameters---numerous new consequences follow, as we outline below.

\subsection{Consequences}

\paragraph{A linear-time approximation algorithm.}
By combining the basic new algorithm with existing (more complicated) techniques,
we show that the $O(\eps^{-2.5} n\log^{1.5}n)$ time bound can be further
reduced
all the way down to \emph{linear}! 
More precisely, the new (randomized) time bound is $O(\eps^{-2}n)$.  Linear-time algorithms were not 
known before, even for the approximate decision problem with a fixed threshold, and even in the binary case ($\sigma=2$).  
In fact, our final time bound is $O(\frac{n\log\sigma}{\log n} +
 \frac{n\log^2\log n}{\eps^2 \log n})$, which is \emph{slightly sublinear} in $n$ when $\sigma$ is small ($\sigma = n^{o(1)}$).

As the reader may surmise,
bit-packing techniques are needed (we assume that the input strings are
given in $O(\frac{n\log\sigma}{\log n})$ words).  To ease the description, in \cref{sec:combo},
we first present a version with $O(\eps^{-2}n\log\log n)$ running time and no messier bit-packing tricks, before describing the final algorithm in \cref{app:lin}.

\paragraph{An improved exact algorithm.}
We apply our linear-time approximation algorithm to obtain a faster algorithm for the \emph{exact} $k$-mismatch problem 
(computing exactly all distances at most $k$).  The new time bound is $O(n + \min(\frac{nk}{\sqrt{m}}\sqrt{\log m},\, \frac{nk^2}{m}))$, which shaves off some logarithmic factors
from Gawrychowski and Uznański's result~\cite{GawUzn} (although to be fair, their result is deterministic).
In particular, the running time is linear when $k\le\sqrt{m}$.
Our description (see \cref{sec:exact}) does not rely on Gawrychowski and Uznański's  and is arguably simpler, using \emph{forward differences}~\cite{CKP19} to handle approximately periodic patterns.

\paragraph{Improved $\eps$-dependence.}
Apart from shaving off $\log n$ factors, our approach, combined with rectangular matrix multiplication (interestingly), leads to approximation algorithms with improved $\eps^{-O(1)}$ factors in the time cost. 
As mentioned, Kopelowitz and Porat~\cite{KopPorFOCS15,KopPor} previously obtained algorithms with a
factor of $\eps^{-1}$, which improve upon earlier methods with an $\eps^{-2}$ factor.
We are able to obtain even better $\eps$-dependence in many cases (see \cref{sec:mm}).
The precise time bound as a function of $m$, $n$, and $\eps$ is tedious to state (as it relies on current results on rectangular
matrix multiplication), but in the case when the pattern is sufficiently long, for example, when $m\ge\eps^{-28}$ (the exponent $28$ has not been optimized), the running time is actually $O(n\polylog n)$ \emph{without any $\eps^{-O(1)}$ factors}, surprisingly!

\paragraph{Sublinear-time algorithms.}
We also show that truly \emph{sublinear-time} (randomized) algorithms are possible for
the approximate decision problem (finding all locations with distances approximately less than $k$) when the threshold $k$ is not too small, the approximation factor is a constant, and the number $\OCC$ of occurrences to report is sublinear.  Such sublinear-time algorithms are
attractive from the perspective of big data, as not all of the input need to be read.  All we assume is that
the input pattern and text are stored in arrays supporting random access.  For example, for an approximation factor $1+\eps$,
we obtain a time bound of $\tO(n/k^{\Omega(\eps^{1/3}/\log^{2/3}(1/\eps))} +
\OCC\cdot k^{O(\eps^{1/3}/\log^{2/3}(1/\eps))})$, and
for an approximation factor near~2, we obtain a time bound of $\OO(n^{4/5} + n/k^{1/4} + \OCC)$.
Different tradeoffs are possible, as the bound relies on known results on approximate  nearest neighbor search in high dimensions (see \cref{sec:sublin}).
The $\OCC$ term disappears if we just want to decide existence or report one location.

In particular, we obtain a \emph{property tester} for pattern matching: with good probability, the test  
returns true if an exact match exists, and false if the pattern is 
\emph{$\delta$-far} from occurring the text, i.e., its
Hamming distance is more than $\delta m$ at every location.
The running time
is $\tO(\delta^{-1/3}n^{2/3} + \delta^{-1}\frac nm)$ (approximate
nearest neighbor search is not needed here, and the algorithm is simple).
We are not aware of such a property tester for pattern matching, despite the extensive literature on property testing and sublinear-time algorithms, and on the classical pattern matching problem.

We remark that some previous work has focused on sublinear-time algorithms with the added assumption that the input strings are generated from a known distribution~\cite{ChangL94,AndoniIKH13}.
By contrast, our results hold for \emph{worst-case} inputs.
Additional work considers sublinear-time algorithms for edit-distance problems~\cite{BatuEKMRRS03,Bar-YossefJKK04,Bar-YossefJKK_APPROX04}.
Nevertheless, some of these sublinear-time algorithms (particularly, by Andoni et al.~\cite{AndoniIKH13} and Batu et al.~\cite{BatuEKMRRS03}) share some rough similarities with our general approach.

\paragraph{Streaming approximation algorithms.}
Yet another setting where our approach leads to new results
is that of (one-pass) \emph{streaming} algorithms.
Characters from the text arrive in a stream one at a time,
and locations with Hamming distance at most $k$ 
need to be identified as soon as their last characters are read.
The goal is to develop algorithms that use limited (sublinear) space, and
also low processing time per character.  Such algorithms are
well-motivated from the perspective of big data.

A breakthrough paper by Porat and Porat~\cite{DBLP:conf/focs/PoratP09} provided a streaming algorithm
for exact pattern matching ($k=0$) working in $\Ohtilde(1)$ space and taking $\Ohtilde(1)$ time per text character
(Breslauer and Galil~\cite{BG:2014} subsequently improved the time cost to $\Oh(1)$).
Porat and Porat~\cite{DBLP:conf/focs/PoratP09} also introduced the first streaming algorithm for the exact $k$-mismatch problem 
using $\Ohtilde(k^2)$ time per character and $\Ohtilde(k^3)$ space. Subsequent improvements~\cite{Cli,GKP18}
culminated in an algorithm by Clifford et al.\ (SODA 2019)~\cite{CKP19} which solves the streaming exact $k$-mismatch problem in $\Ohtilde(\sqrt{k})$ time per character using $\Ohtilde(k)$ space (this space consumption is optimal regardless of the running time).

Streaming algorithms for the approximate $k$-mismatch problem have
also been considered~\cite{Cli}.
However, the only known result not subsumed by the above-mentioned
exact algorithm is by Clifford and Starikovskaya~\cite{DBLP:conf/icalp/CliffordS16}, who gave a streaming algorithm with
$\Ohtilde(\eps^{-5}\sqrt{m})$
space and $\Ohtilde(\eps^{-4})$ time per character, beating
the results for the exact case only when $k\gg\sqrt{m}$.


In \cref{sec:stream}, we describe a streaming algorithm for the approximate $k$-mismatch problem, which is based on our new simple approximation algorithm (\cref{sec:generic}),
with $\Ohtilde(\eps^{-2.5}\sqrt{k})$ space and $\Ohtilde(\eps^{-2.5})$ time per character.  
In \cref{sec:streamingB}, we introduce another sampling approach leading to an algorithm with $\Ohtilde(\eps^{-2}\sqrt{k})$ space and $\Ohtilde(\eps^{-3})$ time per character.  
Moreover, a thorough analysis of our algorithm shows that the space usage is always $\Ohtilde(\eps^{-1.5}\sqrt{m})$.
(Independently, Starikovskaya et al.~\cite{SSU19} apply a different approach to design a streaming algorithm using $\Ohtilde(\eps^{-2}\sqrt{m})$ space.)

\section{Preliminaries}\label{sec:preliminaries}

A string $S$ of length $|S|=s$ is a sequence of characters $S[0]S[1]\cdots S[s-1]$ over an alphabet $\Sigma$. In this work, we assume $\Sigma=[\sigma]$.
A \emph{substring} of $S$ is denoted by $S[i\dd j]=S[i]S[i+1]\cdots S[j]$ for $0\leq i \leq j < s$. If $i=0$, the substring is called a \emph{prefix} of $S$, and if $j=s-1$, the substring is called a \emph{suffix} of $S$.
For two strings $S$ and $S'$ of the same length $|S|=s=|S'|$, we denote by $\HAM(S,S')$ the Hamming distance of $S$ and~$S'$,
that is, $\HAM(S,S') = |\{i\in[s] : S[i] \ne S'[i]\}|$.
Let $\odot$ denote concatenation (in increasing order of the running index).


We begin with a precise statement of the problem in two variants.
We state the problem in a slightly more general form, where we are additionally given a set $Q$ of query locations.  (We may take $Q=[n-m+1]$ at the end
to reproduce the standard formulation.)

\defproblem{prob:apx}{Approximate Text-to-Pattern Hamming Distances}{A pattern $P\in \Sigma^m$, a text $T\in \Sigma^n$, a sorted set $Q\subseteq [n-m+1]$ of query locations, and an error parameter $\eps\in(0,\frac13]$.}{For every $i\in Q$, a value $\DD_i$ such that $(1-\eps)d_i \le \DD_i \le (1+\eps) d_i$,
where $d_i = \HAM(P, T[i\dd i+m-1])$ is the
Hamming distance between $P$  and $T[i\dd i+m-1]$.}

The decision version of the problem, approximately comparing each distance with a given threshold value, is formulated using the notion of an \emph{$(\eps,k)$-estimation}.
We say that $\hat x$ is an $(\eps,k)$-estimation of $x$ if the following holds:
\begin{itemize}
\item if $\tilde{x}\in [(1-\eps)k,2(1+\eps)k]$, then $(1-\eps) x \le \tilde{x} \le (1+\eps)x$;
\item if $\tilde{x} < (1-\eps) k$, then $x < k$;
\item if $\tilde{x} > 2(1+\eps)k$, then $x > 2k$.
\end{itemize}

\defproblem{prob:fixedk}{Approximate Text-to-Pattern Hamming Distances with a Fixed Threshold}{A pattern $P\in \Sigma^m$, a text $T\in \Sigma^n$, a sorted set $Q\subseteq [n-m+1]$ of query locations, a distance threshold $k$, and an error parameter $\eps\in(0,\frac13]$.}{For every $i\in Q$, a value $\DD_i$ that is an $(\eps,k)$-estimation of $d_i$.}

Notice that, for every $k$, a solution for \cref{prob:apx} is also a solution for \cref{prob:fixedk}.
Moreover, given solutions for \cref{prob:fixedk} for each $k$  up to $m$ that is a power of $2$,  \cref{prob:apx} is solved as follows: Let $\DD_i^{(k)}$ be an $(\eps,k)$-estimation of $d_i$.
For every $i\in Q$ and every $k$, if $\DD_i^{(k)}\in [(1-\eps)k,2(1+\eps)k]$ then $\DD_i^{(k)} \in (1\pm\eps) d_i$, and for $k= 2^{\floor{\log d_i}}$ the condition $\DD_i^{(k)}\in [(1-\eps)k,2(1+\eps)k]$ must hold.

%
%

\section{A Generic Sampling Algorithm}\label{sec:generic}
We first focus on \cref{prob:fixedk}. 
We introduce an integer parameter $s>0$ controlling the probability that the algorithm returns correct answers.
For each position $i$, define $M_i := \{j : P[j]\ne T[i+j]\}$ so that $d_i = |M_i|$.
Our algorithm estimates the size $d'_i$ of $M'_i := M_i \bmod p := \{j \bmod p : j \in M_i\}$
for an appropriately chosen integer $p$.
By the following result, if $p$ is a prime number picked uniformly at random from a certain range, then $(1-\eps)d_i\le  d'_i \le d_i$
holds with probability $1-\Oh(1/s)$.
Thus, a good estimation of $d'_i$ is also a good estimation for $d_i$.

\begin{lemma}\label{lem:X}
Let $p$ be a random prime in $[\hat p,2\hat p)$, where $\hat p=\eps^{-1} s k \log m$.
For every set $M\subseteq [m]$ of size $O(k)$, the probability that $|M \bmod p|< (1-\eps)|M|$ is
$O(1/s)$.
\end{lemma}
\begin{proof}
The number of triples $(i,j,p)$ such that $i,j\in M$, $i<j$, and $p\in [\hat p, 2\hat p)$ is a prime divisor
of $j-i$ is at most $O(|M|^2\log_{\hat p} m)$ (since any positive integer in $[m]$ has at most $\log_{\hat p} m$ prime divisors $p\ge \hat p$). If $|M \bmod p|< (1-\eps)|M|$, then the number of such triples
with a fixed prime $p$ is at least $\eps|M|$.
Thus, the number of primes $p\in [\hat p,2\hat p)$ with
$|M \bmod p|< (1-\eps)|M|$ is at most $O(\frac{|M|^2\log_{\hat p} m}{\eps |M|})=O(\eps^{-1} k\log m/\log \hat p)$.  The total number of primes in $[\hat p,2\hat p)$ is
$\Omega(\hat p/\log \hat p)=\Omega(\eps^{-1} s k\log m/\log \hat p)$. Hence, the probability of picking a ``bad'' prime is $O(1/s)$.
\end{proof}

\paragraph{Offset texts and patterns.}
The estimation of  $d'_i$ uses the concept of \emph{offset strings}.
Let $p$ be an integer.
For a string $S$ of length $m$ and an integer $r\in [p]$, we define the $r$th \emph{offset string} as 
\[\bigodot_{j\in [m]:\ j \bmod p \,=\, r  } S[j].\]
Notice that \[M'_i\ =\ \left\{r \in [p] :  \bigodot_{j\in [m]:\ j\bmod p \,=\, r   } P[j]\ \ne  \bigodot_{j\in [m]:\ (i+j) \bmod p \,=\, r   } T[i+j] \right\}.\]

\paragraph{Picking a random offset.}
Unfortunately, finding all occurrences of all offset patterns in all offset texts is too costly.
One way to efficiently estimate $d'_i$ is to randomly pick an offset.
Let $z$ be an integer parameter to be set later such that $1\le z\le p$ and let $b\in [p]$ be an arbitrary integer.
Write $(i\bmod p)$ as $u_i+v_i z$ with $u_i\in [z]$ and $v_i\in [\ceil{p/z}]$.

If the algorithm stores the offset patterns 
\[\bigodot_{j\in [m]:\ (j+u)\bmod p \,=\, b   } P[j]\]
for every $u\in [z]$, and the offset texts 
\[\bigodot_{j\in [m]:\ (i+j-vz)\bmod p \,=\, b   } T[i+j]\]
for every $v\in [\ceil{p/z}]$, then the algorithm has the information needed to test whether $(b-u_i)\bmod p \in M_i'$, i.e., whether 
\[\bigodot_{j\in [m]:\ (j+u_i)\bmod p \,=\, b   } P[j]\ = \bigodot_{j\in [m]:\ (i+j+u_i) \bmod p \,=\, b  } T[i+j]
\ = \bigodot_{j\in [m]:\ (i+j-v_i z) \bmod p \,=\, b  } T[i+j].\]
Moreover, if $b$ is chosen uniformly at random, then $(b-u_i) \bmod p$ is also uniformly random in $[p]$, and so $\Pr[(b-u_i)\bmod p \in M'_i] = \frac {d'_i}{p}$.

\paragraph{Picking multiple random offsets.}
Instead of picking one element $b$, our algorithm picks a random subset of elements $B\subseteq [p]$, with sampling rate $\beta=\frac{1}{2k}$.
(The expected size of $B$ is small, namely, $O(\beta p)=O(\eps^{-1}s\log m)$, if $p=\Theta(\eps^{-1}sk\log m)$ and $s$ is small.) 
For each $i$, let $E_i$ be the event that there exists some $b\in B$ such that $(b-u_i)\bmod p \in M'_i$.

\begin{lemma}\label{lem:main}
	$\Pr[E_i] = 1-(1-\beta)^{d'_i}$.
	\end{lemma}
	\begin{proof}
		Note that $E_i$ holds if and only if $B'_i \cap M'_i \ne \emptyset$,
		where $B'_i := \{(b- u_i)\bmod p : b\in B\}$.
	As $B'_i$ is a subset of $[p]$ with each element sampled independently with rate $\beta$,
	$\Pr[B'_i \cap M'_i = \emptyset] = (1-\beta)^{|M'_i|}$.
\end{proof}

Our algorithm tests for each location $i$ whether $E_i$ happens.
This is equivalent to testing if
\[\bigodot_{j\in [m]:\ (j+u_i)\bmod p \,\in\, B   } P[j]\ = \bigodot_{j\in [m]:\ (i+j-v_iz)\bmod p \,\in\, B  } T[i+j].\]

Finally, in order to extract an estimation of $d'_i$, the algorithm repeats the process with $L=\eps^{-2} \log s$ independent choices of $B$.
For each location $i\in Q$, the algorithm computes $c_i$ which is the overall number of times that the event $E_i$ took place throughout the $L$ executions.
Finally, the algorithm sets $\DD_i = \log_{1-\beta}(1-c_i / L)$ (so that
$c_i = (1-(1-\beta)^{\DD_i})\cdot L$).
The following pseudo-code summarizes the generic sampling algorithm,
whose correctness follows from \cref{lem:X,lem:main}, and a standard use of Chernoff bounds.

\begin{algorithm}[H]
	Pick a random prime $p\in [\hat p,2 \hat p)$\tcr*{$\hat{p} = \eps^{-1} s k \log m$}\label{ln:pick}
	$p = \min(p,m)$\;\label{ln:m}
	\ForEach(\tcr*[f]{$L = \Theta(\eps^{-2}\log s)$ with a sufficiently large constant factor}){$\ell\in [L]$}{
		Pick a random sample $B^{(\ell)}\subseteq [p]$ with sampling rate $\beta$\tcr*{$\beta = \frac{1}{2k}$}\label{ln:sample}
		\lForEach(\tcr*[f]{$1 \le z \le p$}){$u \in [z]$} {
			$\displaystyle X_u^{(\ell)}\: =\: \bigodot_{j\in [m]:\ (j+u) \bmod p\;\in\; B^{(\ell)}} P[j]$\label{ln:constructX}
		}
		\ForEach{$v \in [\ceil{p/z}]$}{
			\lForEach{$i \in [n-m+1]$} {
				$\displaystyle Y_v^{(\ell)}(i) \:=\: \bigodot_{j\in [m]:\ (i+j-vz) \bmod p \;\in\; B^{(\ell)} } T[i+j]$\label{ln:constructY}
			}
		}
	}
	\ForEach{$i \in Q$}{\label{ln:query-generic}
		Write $(i\bmod p)$ as $u_i+v_iz$ with $u_i\in [z]$ and $v_i\in [\ceil{p/z}]$\;
		Set $c_i=|\{\ell\in [L]: X_{u_i}^{(\ell)}\neq Y_{v_i}^{(\ell)}(i)\}|$ and $\DD_i = \log_{1-\beta}(1-c_i / L)$\;\label{ln:output}
	}
	\caption{Generic-Algorithm($T,P,Q,k,\eps,s$)}\label{alg:main}
\end{algorithm}

\begin{restatable}{lemma}{lemgeneric}\label{lem:generic_algo_correct}
For every $i\in Q$, the value $\DD_i$ computed by \cref{alg:main} is an $(\eps,k)$-estimation of $d_i$ with probability $1-\Oh(1/s)$. 
\end{restatable}
\begin{proof}
Let $E_i^{(\ell)}$ be the event that there exists some $b\in B^{(\ell)}$ such that $(b-u_i)\bmod p \in M'_i$.
By \cref{lem:main}, we have $\Pr[E_i^{(\ell)}] = 1-(1-\beta)^{d'_i}$ and the events $E_i^{(\ell)}$ are independent across $\ell \in [L]$.
Let $\teps = \Theta(\eps)$ with a sufficiently small constant factor.
The symmetric multiplicative Chernoff bound therefore yields
\[\Pr \left[c_i / L \in (1\pm \teps)\left(1-(1-\beta)^{d'_i}\right)\right] \,=\, 1 - \exp\left(-\Omega\left(\teps^2 L\left(1-(1-\beta)^{d'_i}\right)\right)\right).\]
We consider three cases.

\paragraph*{Case 1: $d_i \in [\frac12 k, 4k]$.}
By \cref{lem:X}, $(1-\teps) d_i \le |M_i \bmod p| \le d_i$ holds for any $\teps=\Theta(\eps)$ with probability $1-O(1/s)$ for the prime $p$ picked in \cref{ln:pick}, and obviously this is also true if $p$ is replaced with $m$ in \cref{ln:m}. The following argument is conditioned on that event.
In other words, we assume that $(1-\teps) d_i \le d'_i \le d_i$.
In particular, this yields  $\frac{1-\teps}{2}k \le d'_i \le 4k$. Since $\beta = \frac{1}{2k}$ for $k\ge 1$,
we have $1-\beta = \exp(-\Theta(\frac{1}{k}))$, and thus  $(1-\beta)^{d'_i} = \exp(-\Theta(1))$.

The Chernoff bound therefore yields that
\[c_i / L\,\in\, (1\pm \teps)\left(1-(1-\beta)^{d'_i}\right)\,=\,1-(1-\beta)^{d'_i}(1\pm O(\teps))\]
holds with probability \[1-\exp(-\Omega(\teps^2 L(1-(1-\beta)^{d'_i})))\,=\,1-\exp(-\Omega(\teps^2 L))\,=\,1-\tfrac{1}{s}\]
provided that the constant factor at $L=\Theta(\teps^{-2}\log s)=\Theta(\eps^{-2}\log s)$ is sufficiently large.
Hence,
\[
    \DD_i = \log_{1-\beta}((1-\beta)^{d'_i}(1\pm O(\teps)))=d'_i+\log_{1-\beta}(1\pm O(\teps))=d'_i\pm O(\teps k)=d_i\pm O(\teps k)=d_i(1\pm O(\teps))
\]
holds with probability $1-\Oh(1/s)$. This remains true even if we account for the fact that $(1-\teps) d_i \le d'_i \le d_i$ may fail to be satisfied with probability $O(1/s)$. If the constant factor at $\teps = \Theta(\eps)$ is sufficiently small, we conclude that
$(1-\eps) d_i \le \DD_i \le (1+\eps)d_i$ holds in this case with probability $1-\Oh(1/s)$.
In particular, $d_i < k$ if $\DD_i < (1-\eps)k$ and $d_i > 2k$ if $\DD_i > 2(1+\eps)k$, so $\DD_i$ is an $(\eps,k)$-estimation of $d_i$.

\paragraph*{Case 2: $d_i < \frac12 k$.}
In this case, $d'_i < \frac12 k$ and thus $(1-\beta)^{d'_i} = \exp(-O(1))$.
The Chernoff bound therefore yields that
\[c_i / L\,\le\, (1+\teps)\left(1-(1-\beta)^{k/2}\right)\,=\,1-(1-\beta)^{k/2}(1- O(\teps))\] holds with probability
\[1-\exp(-\Omega(\teps^2 L(1-(1-\beta)^{k/2})))\,=\,1-\exp(-\Omega(\eps^2 L))\,=\,1-\tfrac{1}{s}\]
provided that the constant factor at $L=\Theta(\teps^{-2}\log s)=\Theta(\eps^{-2}\log s)$ is sufficiently large.
Hence,
\[   \DD_i \,\le\, \log_{1-\beta}((1-\beta)^{k/2}(1- O(\teps)))\,=\,k/2+\log_{1-\beta}(1- O(\teps))\,=\,k/2+ O(\teps k)\,=\,k/2(1+ O(\teps)) \] holds with probability
$1-\Oh(1/s)$.
If the constant factor at $\teps = \Theta(\eps)$ is sufficiently small, we conclude that
$\DD_i \le (1+\eps)\frac12k$ holds in this case with probability $1-\Oh(1/s)$.
Since $\eps \le \frac13$, this implies $\DD_i < (1-\eps) k$ and hence $\DD_i$ is an $(\eps,k)$-estimation of $d_i$.

\paragraph*{Case 3: $d_i > 4k$.}
\cref{lem:X} applied to any fixed subset of $M_i$ of size $4k$ implies that $d'_i > (1-\teps)4k$ holds with probability $1-O(1/s)$.
The following argument is conditioned on that event.
The Chernoff bound therefore yields that
\[c_i / L\,\ge\, (1-\teps)\left(1-(1-\beta)^{4(1-\teps)k}\right)\,=\,1-(1-\beta)^{4(1-\teps)k}(1+ O(\teps))\] holds with probability
\[1-\exp(-\Omega(\teps^2 L(1-(1-\beta)^{4(1-\teps)k})))\,=\,1-\exp(-\Omega(\teps^2 L))\,=\,1-\tfrac{1}{s}\]
provided that the constant factor at $L=\Theta(\teps^{-2}\log s)=\Theta(\eps^{-2}\log s)$ is sufficiently large.
Hence,
\[
    \DD_i \ge \log_{1-\beta}((1-\beta)^{4(1-\teps)k}(1+ O(\teps)))=4(1-\teps)k+\log_{1-\beta}(1+ O(\teps))=4(1-\teps)k- O(\teps k)=4k(1-O(\teps))
\]
holds with probability $1-\Oh(1/s)$. This remains true even if we account for the fact that $(1-\teps) d_i \le d'_i \le d_i$ may fail to be satisfied with probability $O(1/s)$. If the constant factor at $\teps = \Theta(\eps)$ is sufficiently small, we conclude that
$\DD_i \ge 4(1-\eps)k$ holds in this case with probability $1-\Oh(1/s)$. Since $\eps \le \frac13$, this implies  $\DD_i > 2(1+\eps) k$ and hence $\DD_i$ is an $(\eps,k)$-estimation of $d_i$.
\end{proof}

\section{A Simple $O(n\log^{1.5}n)$-Time Implementation}\label{sec:offline}

We now describe a simple implementation of  \cref{alg:main}, with $O_\eps(n\log^{1.5}n)$ running time.
Our algorithm uses a standard family of Karp--Rabin-style
fingerprint functions~\cite{KarRab,MotRag}, which is summarized in the following lemma.

\begin{lemma}\label{lem:fingerprint}
\emph{(Fingerprint functions)}
Given a prime number $q\ge\sigma$, define $\FF_q=\{F_{x,q}:x\in [q]\}$ where
$F_{x,q}:\Sigma^*\rightarrow[q]$ is the function
$F_{x,q}(S)=(\sum_{i=0}^{|S|-1} S[i]x^i)\bmod q$.
For a random $F\in\FF_q$ and fixed length-$m$ strings $X$ and $Y$ with $X\neq Y$, we have $\Pr[F(X)=F(Y)]\le \frac{m}{q}$.
\end{lemma}
\begin{proof}
We have
$F_{x,q}(X)=F_{x,q}(Y)$ if and only if $x$ is one of the at most $m$ roots
of the polynomial
$\sum_{i=0}^{m-1} (X[i]-Y[i])x^i$ modulo $q$.
\end{proof}
Our algorithm also applies
the following known family of hash functions mapping $[u]$ to $\{0,1\}$:

\begin{lemma}\label{lem:01}
\emph{(Strong universal hash functions into $\{0,1\}$)}
Define $\HH_u=\{h_{x,u} : x\in [2^{\lceil\log u\rceil}]\}$,
 where $h_{x,u}:[u]\rightarrow\{0,1\}$ is the function with
$h_{x,u}(a)=\bigoplus_{i=0}^{\ell-1}a_ix_i$, where $\oplus$ denotes exclusive-or, $\ell=\lceil\log u\rceil$, and $a_{\ell-1}\cdots a_0$ and $x_{\ell-1}\cdots x_0$ are the
binary representations of $a$ and $x$.
For a random function $h\in \HH_u$ and
fixed numbers $a,b\in [u]$ with $a\neq b$,
we have $\Pr[h(a)=h(b)]=\tfrac12$.
\end{lemma}
\begin{proof}
Suppose that $a$ and $b$ have binary representations $a_{\ell-1}\cdots a_0$ and $b_{\ell-1}\cdots b_0$.  For a random $x\in[2^\ell]$ with binary representation $x_{\ell-1}\cdots x_0$, we have
$h_{x,u}(a)=h_{x,u}(b)$ if and only if $\bigoplus_{k:\ a_{k}\neq b_{k}} x_{k}=0$, which holds with probability exactly $\tfrac12$.
\end{proof}

The following theorem gives a solution to
\cref{prob:fixedk}.
What is notable about the time bound  below is that the first two terms are \emph{sublinear} in many cases: when the threshold $k$ is not too small (and when we choose a small $s$), the algorithm only needs to read a sublinear number of symbols from the text and pattern.

\begin{theorem}\label{thm:main}
For every $s= n^{O(1)}$, 
\cref{prob:fixedk} can be solved in time
\[O\left(\sqrt{\tfrac{s n m\log m}{\eps^5 k}}\log s \,+\, \tfrac{n \log s}{\eps^2 k} \,+\, \tfrac{1}{\eps^2}|Q|\right)\]
using a randomized algorithm whose error probability for each fixed $i\in Q$ is $\Oh(1/s)$.
\end{theorem}

\begin{proof}
Our solution implements \cref{alg:main}.
Thus, by \cref{lem:generic_algo_correct}, the algorithm solves \cref{prob:fixedk} with the desired probability.
Notice that we do not need to recompute the strings $Y_v^{(\ell)}(i)$ for every location~$i$. This is because
the only changes that may need to be made to $Y_v^{(\ell)}(i)$ as $i$ increments are appending $T[i+m]$ to the end and dropping $T[i]$ from the beginning.
To support fast comparisons in \cref{ln:output}, the strings $X_u^{(\ell)}$ and $Y_v^{(\ell)}(i)$ are not stored explicitly, but rather are represented by fingerprints $F^{(\ell)}(X_u^{(\ell)})$ and $F^{(\ell)}(Y_v^{(\ell)}(i))$ for a random function $F^{(\ell)} \in \FF_q$,
where $q$ is some large enough prime.
Thus, the comparisons in \cref{ln:output} are implemented in $O(1)$ time per comparison, and with appropriate $q=n^{O(1)}$, the comparisons are correct with probability $1-\Oh(1/s)$.

\paragraph{Analysis of running time.}
By a Chernoff bound, the total size of sets $B^{(\ell)}$ is $\Theta(\beta L p)$ with probability $1-\exp(\Omega(\beta L p))=1-\Oh(1/s)$
provided that the constant factor at $L=\Theta(\teps^{-2}\log s)=\Theta(\eps^{-2}\log s)$ is sufficiently large.
The analysis below is conditioned on that event.
Thus, \cref{ln:sample} takes $O(\beta Lp)=O(\beta Lm)$ time in total if an efficient sampling algorithm is used~\cite{DBLP:conf/icalp/BringmannF13,DBLP:journals/algorithmica/BringmannP17}.

\cref{ln:constructX} requires computing the fingerprint of a string of length $O(\tfrac{m}{p}|B^{(\ell)}|)$
and costs $O(\tfrac{m}{p} |B^{(\ell)}|)$ time per $u$ and $\ell$.
The total cost over all $u\in [z]$ and all $\ell\in [L]$ is
$O(\beta m z L)$ time.
\cref{ln:constructY}  requires computing the fingerprints of sliding windows
over a string of length $O(\tfrac{n}{p} |B^{(\ell)}|)$ and takes $O(\tfrac{n}{p} |B^{(\ell)}|)$ time per $v$ and $\ell$ 
(since the fingerprint of each window can be computed in constant time from the fingerprint of the previous window).
The total cost over all $v\in [\ceil{p/z}]$ and all $\ell\in [L]$ is
$O(\tfrac{n}{p}\cdot \beta L p \cdot \tfrac{p}{z})=O(\beta n \cdot (p/z) \cdot L)$ time.

The total time bound so far is
$O(\beta\cdot (mz + np/z)\cdot L)$.
Setting $z=\min(\lfloor\sqrt{np/m}\rfloor,\, p)$ gives
\[O\left(\beta\cdot\left(\sqrt{nmp}+n\right)\cdot L\right)
=O\left(\tfrac1k\left(\sqrt{nm\eps^{-1}sk \log m}+n\right)\tfrac{\log s}{\eps^2}\right)= O\left(\sqrt{\tfrac{snm \log m}{\eps^5 k}}\log s + \tfrac{n\log s}{\eps^2 k}\right).
\]

In Lines~\ref{ln:query-generic}--\ref{ln:output}, the algorithm examines the indices $i\in Q$ in increasing order.  At any time, the algorithm maintains a pointer to a previous value of $F^{(\ell)}(Y_v^{(\ell)}(i))$ for each $v$ and $\ell$.
As the algorithm examines the next $i\in Q$, it advances $L$ pointers to obtain the current values of $F^{(\ell)}(Y_v^{(\ell)}(i))$ for all $\ell$.
The total cost for advancing pointers is $O(\beta n\cdot (p/z)\cdot L)$, which is already accounted for.
In addition, in \cref{ln:output} the algorithm spends $O(L)$ time per $i\in Q$,
for a total of $O(|Q| L)=O(\eps^{-2} |Q|\log s)$ time.

\paragraph{Speed-up by bit packing.}
We describe a simple improvement to reduce the running time of Lines~\ref{ln:query-generic}--\ref{ln:output} from $O(\eps^{-2}|Q|\log s)$ to $O(\eps^{-2} |Q|)$.  We work in the word RAM model with $w$-bit words, where $w=\delta\log n$ for a sufficiently
small constant $\delta$.

First, we change the fingerprint functions.
At each iteration $\ell$, the algorithm additionally picks a random hash function $h^{(\ell)}\in \HH_M$ and replaces $F^{(\ell)}$ with $h^{(\ell)}\circ F^{(\ell)}$.  Note that $h^{(\ell)}$ can be evaluated in $O(1)$ word operations.
Let $x^{(\ell)}_u = h^{(\ell)}(F^{(\ell)}(X_u^{(\ell)}))$ and $y^{(\ell)}_v(i) = h^{(\ell)}(F^{(\ell)}(Y_v^{(\ell)}(i)))$.
By \cref{lem:01}, for each $i\in Q$, $\Pr[x^{(\ell)}_{u_i} = y^{(\ell)}_{v_i}(i)] = \tfrac12 \Pr[E^{(\ell)}_i]$.
The algorithm doubles $c_i$ to compensate.

For each $u\in [z]$, the algorithm stores $\vec{x}_u=\langle x_u^{(\ell)} :\ell\in [L]\rangle$ as a bit vector packed in $O(\ceil{L/w})$ words.
As $i$ increases, the algorithm maintains the current $\vec{y}_v(i)=\langle y_v^{(\ell)}(i) :\ell\in [L]\rangle$ stored as a bit vector packed in $O(\ceil{L/w})$ words for every $v\in [\ceil{p/z}]$.
The update cost is proportional to
the number of changes to $\vec{y}_v(i)$ as $i$ increases in Lines~\ref{ln:query-generic}--\ref{ln:output}.
The number of such changes is $O(\beta n\cdot (p/z)\cdot L)$.
Note that the algorithm  can pre-sort the indices $i$ at which the changes occur,
for example, by a 2-pass radix sort with an $O(\sqrt{n})$-time overhead.
\cref{ln:output} can then be executed by looking up
the bit vectors $\vec{x}_{u_i}$ and $\vec{y}_{v_i}(i)$
and applying $O(\ceil{L/w})=O(\eps^{-2})$ word operations per $i\in Q$.
The total time cost is $O(\eps^{-2}|Q|)$.

We have assumed the following word operations are available: (i)~bitwise-xor and (ii)~counting the number of 1-bits in a word.  If these operations are not directly supported, they can still be implemented in constant time by lookup in a table of size
$2^w=n^\delta$.
\end{proof}


As an immediate consequence, we get the following worst-case time bound, which already improves the previous $O_\eps (n\log^2n)$ bound as a function of $n$.

\begin{corollary}\label{cor:main}
There is a randomized algorithm solving \cref{prob:apx} in $O(\eps^{-2.5}n\log^{1.5}n)$ time, returning answers correct with high probability. 
\end{corollary}
\begin{proof}
We run the algorithm for a sufficiently large constant~$s$ (in this application, the simpler version without bit packing suffices), and
repeat $O(c\log n)$ times (taking the median of the
answers for each $i\in Q$) to lower the error probability per $i\in Q$ to
$O(n^{-c-1})$.  This solves \cref{prob:fixedk} in time
\[O\left(\sqrt{\tfrac{nm \log m}{\eps^5 k}}\log n + \tfrac{n \log n}{\eps^{2} k} + \tfrac{|Q|\log n}{\eps^{2}}\right).\]
Notice that the algorithm developed in \cref{thm:main} supports processing locations $i\in Q$ online (as long as they are provided in the increasing order).
Hence, we run $\Oh(\log m)$ instances of this algorithm in parallel, one for each power of two $k \le m$.

For each $i\in Q$, the algorithm performs a binary search over the $O(\log m)$ powers of two, which results in forwarding $i$ to $O(\log \log m)$ 
out of the $O(\log m)$ instances of the algorithm of \cref{thm:main}.

The overall running time is therefore
\[ O\Bigg(\sum_{k} \sqrt{\tfrac{nm \log m}{\eps^5 k}}\log n + \sum_{k} \tfrac{n \log n}{\eps^{2} k} + \tfrac{|Q|\log n \log \log m}{\eps^{2}}\Bigg).\]
Since the first two terms are geometric progressions, the time cost becomes  
\[O\Big(\sqrt{\tfrac{nm \log m}{\eps^5}}\log n +\tfrac{n \log n}{\eps^{2}} +\tfrac{ |Q|\log n \log \log m}{\eps^{2}}\Big) = O(\eps^{-2.5} n \log^{1.5} n).\qedhere\] 
\end{proof}

\section{Further Consequences: An Overview}\label{sec:sketch}

Our approach leads to many further consequences, in many cases, by careful re-implementations of our generic algorithm.  We give a rough overview in this section, and defer detailed proofs to
subsequent sections.

\paragraph{Towards a linear-time approximation algorithm.}
We first note that the $O_\eps(n\log^{1.5}n)$ upper bound in \cref{cor:main}
is an overestimate when $k$ is large: from the proof, we see that
the total running time is actually at most
\[O_\eps\left(\sqrt{\tfrac{nm \log m}{k}}\log n + \tfrac{n\log n}{k} + |Q|\log n\log\log m\right)
\:=\: O_\eps\left(\tfrac{n \log^{1.5}n }{\sqrt{k}}+ n\log n\log\log m\right).\]
On the other hand, when $k$ is small, e.g., $k\le \log n$, we can switch to a known
exact algorithm, e.g., with $O(n\sqrt{k\log k})$ running time~\cite{Ami} (although this requires FFT\@).
The minimum of the two already yields an improved time bound of $O(n\log n\log\log n)$ 
for \cref{prob:apx}.

To do still better, we combine three algorithms:

\begin{itemize}
\item
Case I: $m$ is small, e.g., $m\le \log^{O(1)}n$.  In this case, \cref{prob:fixedk} can be solved in linear time by a simplification of our algorithm, as we show in \cref{sec:small} (see \cref{thm:small}).
\item
Case II: $k$ is small, e.g., $k\le m^\delta$ for some constant $\delta$.  In this case, we can
switch to a known exact algorithm, e.g., one by Cole and Hariharan~\cite{ColHar}, with running time of $O(n + \frac{nk^4}{m})$, which is linear for $\delta<\tfrac14$. 
Having been designed primarily for pattern matching with respect to edit distance, Cole and Hariharan's algorithm is quite complicated and inefficient (in terms of the polynomial dependence on $k$).  To be more self-contained, we describe an exact algorithm in \cref{sec:exact} (see \cref{thm:exact}), which actually has a better running time of $O(n + \frac{nk^2}{m})$.  (This does not require FFT.)
\item
Case III: $k>m^\delta$ and $m>\log^{\omega(1)}n$.
Here, we go back to our algorithm in \cref{sec:offline}, but with $s=n^{\delta/2}$, to solve \cref{prob:fixedk}.  The running time is
$O_\eps\left(\sqrt{\tfrac{s nm \log m}{k}}\log s + \tfrac{n \log s}{k} + |Q|\right)
=O_\eps( \sqrt{snm^{1-\delta}}\log^{1.5}n + \tfrac{n\log n}{m^\delta} + n)=O_\eps(n)$.
The error probability $O(n^{-\delta/2})$ can be lowered by a constant number of repetitions.
\end{itemize}
In all cases, we thus obtain a linear-time approximation algorithm for \cref{prob:fixedk}.  The algorithm can be modified to solve
\cref{prob:apx}, though the running time increases
to $O_\eps(n\log\log n)$ (see \cref{cor:combo}).  To remove the $\log\log n$ factor,
we additionally use bit-packing tricks to reimplement the algorithms in all three cases.
This, in fact, leads to a slightly sublinear time bound of
$O(\frac{n\log\sigma}{\log n} + \frac{n\log^2\log n}{\eps^2 \log n})$;
the details are more complicated and are deferred to \cref{app:lin}.

\paragraph{Improved $\eps$-dependence, via rectangular matrix multiplication.}
By a different implementation, it is possible to obtain $O(n\polylog n)$ running time without any $\eps^{-O(1)}$ factor when the pattern is long enough, namely, when $m\ge \eps^{-c}$ for some sufficiently large constant $c$.  First, we may assume that $k\ge\sqrt{m}\ge \eps^{-c/2}$, for otherwise we can switch to an exact $\tO(n+\frac{nk^2}{m})$-time algorithm.

Our algorithm in \cref{sec:offline}, with $s=O(1)$, has running time
$\tO\left(\sqrt{\tfrac{m n}{\eps^5 k}} + \tfrac{n}{\eps^2 k}
+ \tfrac{1}{\eps^2}|Q|\right)$.  Notice that the $\eps^{-O(1)}$
factors in the first two terms disappear when $k$ is large.  The third term 
comes from Lines~\ref{ln:query-generic}--\ref{ln:output}, i.e., the computation of the counts $c_i$, 
which amounts to the computation of inner products between vectors $\vec{x}_i$ and $\vec{y}_i(i)$.  The vectors have dimension $L=O(\eps^{-2})$.  There are $O(z)$ different vectors $\vec{x}_i$, and it is not difficult to show that there are $O(\frac{n}{z} + \frac{n}{\eps^2k})$ different vectors $\vec{y}_i(i)$ (in expectation).
Therefore, this step reduces to the multiplication of an $O(z)\times O(\eps^{-2})$
matrix and an $O(\eps^{-2})\times O(\frac{n}{z} + \frac{n}{\eps^2k})$ matrix.  For $k$ (and thus $m$) sufficiently large, and for an appropriate choice of $z$,
known rectangular matrix multiplication algorithms~\cite{DBLP:journals/siamcomp/Coppersmith82} take time
near linear in the number of output entries $\tO(z\cdot (\frac{n}{z} + \frac{n}{\eps^2k}))=\tO(n)$.  See \cref{sec:mm} for the details.

\paragraph{Sublinear-time algorithms, via approximate nearest neighbors.}
When $k$ is not too small ($n^{\Omega(1)}$) and the approximation factor is a constant, it is possible to obtain truly
sublinear-time algorithms for finding locations with Hamming distance approximately at most $k$ (assuming that
the number of occurrences to report is sublinear).

Recall that the algorithm in \cref{sec:offline} has running time
$\tO_\eps\left(\sqrt{\tfrac{sm n}{k}} + \tfrac{n}{k}
+ |Q|\right)$.  Notice that the first two terms are already sublinear when $k$ is large.  Again, the third term is the bottleneck, coming from Lines~\ref{ln:query-generic}--\ref{ln:output}, i.e., the computation of the counts $c_i$, which correspond to Hamming distances between vectors $\vec{x}_i$ and $\vec{y}_i(i)$.  We can no longer afford to loop through all indices $i$, but 
we just want to identify all $i$ for which $c_i$ is approximately less than some threshold value.  This step reduces to reporting close pairs between 
a set of $O(z)$ vectors and a set of $O(\frac{n}{z} + \frac{n}{\eps^2k})$ vectors.  This subproblem can be solved by using known techniques for (offline) approximate Hamming nearest neighbor search~\cite{AndoniLRW17,AlmanCW20}.

Two technical issues arise.  First, not all pairs of vectors should be matched (i.e., correspond to a valid index $i$).  However, we can identify which vectors $\vec{x}_i$ to match with each $\vec{y}_i(i)$, and these vectors form a contiguous subsequence of $\vec{x}_0,\ldots,\vec{x}_{z-1}$. 
Second, there will be false positives---$O(\frac{n}{s})$ of them in expectation, since
the error probability per position is $O(1/s)$.  However, we can still choose the parameter $s$ to keep all terms sublinear.  See \cref{sec:sublin} for the details.

\paragraph{Streaming algorithms, via multi-stream dictionary matching.}
In the streaming model, we re-implement our generic algorithm differently by treating each $Y^{(\ell)}_v$ as a stream.  Computing the count $c_i$ reduces to exact matching of the pattern $X^{(\ell)}_{u_i}$ in the stream $Y^{(\ell)}_{v_i}$ for each $\ell$.  To this end, we could use a known streaming algorithm for pattern matching.  However, because there are $O(z)$ possible $u_i$'s and $O(p/z)$ possible $v_i$'s, we actually need a streaming pattern matching algorithm that can handle multiple patterns and multiple text streams---luckily, this variant, known as multi-stream dictionary matching, has already been addressed in a recent paper by Golan et al.~\cite{GKP18}.  The space bound is $\tO_\eps(z + p/z)$, which becomes $\tO_\eps(\sqrt{p})=\tO_\eps(\sqrt{k})$ by setting $z=\sqrt{p}$,
and the per-character running time is $\Oh_\eps(1)$. See \cref{sec:stream} for more details and
\cref{sec:streamingB} for an alternative streaming algorithm with improved dependence on $\eps$ in the space consumption.

\section{Simplified Algorithm for Small $m$}\label{sec:small}

In this section, we note that our algorithm in \cref{sec:offline} becomes quite simple if $m\le \sqrt{n}$.  
As this case will be useful later, we provide a self-contained description
of the simplified algorithm below:

\begin{theorem}\label{thm:small}
For every $s=n^{O(1)}$, there is a randomized algorithm for \cref{prob:fixedk}
with running time $O(\eps^{-2}(m^2\log s + n))$,
where the error probability for each fixed $i\in Q$ is $\Oh(1/s)$.
\end{theorem}
\begin{proof}
Our solution is presented as \cref{alg:simple}. Compared to \cref{alg:main},
we set $p=m$ (the analysis involving primes becomes unnecessary\@!) and $z=m$ (so that the sample $B$ is considered with all $m$ shifts in the pattern
but with just one shift in the text). Furthermore, the Karp--Rabin fingerprints are removed, with strings directly hashed to $\{0,1\}$
using \cref{lem:01}.

\begin{algorithm}[H]
	\ForEach(\tcr*[f]{$L = \eps^{-2}\log s$}){$\ell\in [L]$}{
		Pick a random sample $B^{(\ell)}\subseteq [m]$ with sampling rate $\beta$\tcr*{$\beta = \frac{1}{2k}$}\label{ln:sample:simple}
		\lForEach{$i\in [m]$}{Pick a uniformly random function $h^{(\ell)}_i:[\sigma]\rightarrow\{0,1\}$}\label{ln:hash:simple}
		\lForEach{$i \in [m]$} {
			$\displaystyle x_i^{(\ell)}\: =\: \bigoplus_{j\in [m]:\ (i+j) \bmod m\;\in\; B^{(\ell)}} h^{(\ell)}_{(i+j)\bmod m}(P[j])$\label{ln:constructX:simple}%
			}
		\lForEach{$i \in [n-m+1]$}{
				$\displaystyle y_i^{(\ell)} \:=\: \bigoplus_{j\in [m]:\ (i+j) \bmod m \;\in\; B^{(\ell)} } h^{(\ell)}_{(i+j)\bmod m}(T[i+j])$\label{ln:constructY:simple}%
		}
	}
	\lForEach{$i \in Q$}{
		Set $c_i=|\{\ell\in [L]: x_{i \log m}^{(\ell)}\neq y_{i}^{(\ell)}\}|$ and $\DD_i = \log_{1-\beta}(1-2c_i / L)$
	}\label{ln:compare:simple}
	\caption{Simple-Algorithm($T,P,Q,k,\eps,s$)}\label{alg:simple}
\end{algorithm}

\paragraph{Analysis of error probability.}
Recall that $M_i = \{j \in [m]: P[j] \ne T[i+j]\}$ is of size $d_i$.
Define $E^{(\ell)}_{i}$ as the event that $(i+j)\bmod m \in B^{(\ell)}$ for some $j\in M_i$.
Observe that $E^{(\ell)}_{i}$ holds if and only if
\[\bigodot_{j\in [m]:\ (i+j) \bmod m\;\in\; B^{(\ell)}} P[j]\ \ne \bigodot_{j\in [m]:\ (i+j) \bmod m\;\in\; B^{(\ell)}} T[i+j].  \]
On the other hand, the construction of $B^{(\ell)}$ assures that $\Pr[E^{(\ell)}_{i}]=1-(1-\beta)^{d_i}$ analogously to \cref{lem:main}.
Moreover, by \cref{lem:01}, if $E^{(\ell)}_{i}$ holds then $\Pr[x_i^{(\ell)} = y_i^{(\ell)}]=\frac12$.
Otherwise, obviously $\Pr[x_i^{(\ell)} = y_i^{(\ell)}]=1$.
Hence, $\Pr[x_i^{(\ell)} \ne y_i^{(\ell)}] = \tfrac12 \Pr[E^{(\ell)}_{i}]= \tfrac12 \left(1-(1-\beta)^{d_i}\right)$.
Repeating the proof of \cref{lem:generic_algo_correct} (simplified accordingly due to $p=m$), we obtain the following result:
\begin{lemma}\label{lem:generic:simple}
	For every $i\in Q$, the value $\DD_i$ computed by \cref{alg:simple} is an $(\eps,k)$-estimation of $d_i$ with probability $1-\Oh(1/s)$. 
\end{lemma}

\paragraph{Analysis of running time.}
\cref{ln:sample:simple,ln:hash:simple} take $O(m\sigma)$ time per $\ell$, for a total of $O(m\sigma L)\le O(m^2 L)=O(\eps^{-2} m^2\log s )$ time.
\cref{ln:constructX:simple} takes $O(m)$ time per $i\in [m]$ and $\ell$, for a total of
$O(m^2L)=O(\eps^{-2} m^2\log s )$ time.

Implemented using a sliding window, \cref{ln:constructY:simple} takes $O(n)$ time per $\ell$, for a total of $O(nL)=O(\eps^{-2} n\log s )$.
\cref{ln:compare:simple} takes $O(L)$ time per $i\in Q$, for a total of $O(nL)$ as well.
Next, we use bit packing to speedup these steps.

For each $i\in [m]$, we store $\vec{x}_i=\langle x_i^{(\ell)}:\ell\in [L]\rangle$ as a bit vector packed in $O(\ceil{L/w})$ words.
For each $i\in [m]$ and $a\in [\sigma]$, we also
store  a bit vector  $\vec{h}_{i,a} = \langle h_{i,a,\ell}:\ell\in [L]\rangle$,
where $h_{i,a,\ell}=0$ if $i\notin B^{(\ell)}$ and $h_{i,a,\ell}=h^{(\ell)}_{i}(a)$ otherwise.
Then, to compute
the bit vector $\vec{y}_i=\langle y_i^{(\ell)}:\ell\in [L]\rangle$ in \cref{ln:constructY:simple},
we can take the bitwise exclusive-or of the vectors
$\vec{y}_{i-1}$, $\vec{h}_{(i-1)\bmod m\,,\,T[i-1]} $
and $\vec{h}_{(i+m-1)\bmod m\,,\, T[i+m-1]}$, in
$O(\ceil{L/w})=O(\eps^{-2})$ time per $i\in [n-m+1]$.
The total time is $O(\eps^{-2} n)$.
\cref{ln:compare:simple} also takes $O(\ceil{L/w})=O(\eps^{-2})$ time per $i\in Q$,
for a total of $O(\eps^{-2}|Q|)$ time.
\end{proof}

The following result is obtained
by combining \cref{thm:main} with the simpler (and slightly more efficient) approach for the case when $m \ll n$.

\begin{restatable}{theorem}{thmfilter}\label{thm:filter}
	For every constant $\delta > 0$, there is a randomized algorithm for \cref{prob:fixedk} with $k\ge \eps^{-1}m^{\delta}$ that runs in $\Oh(\eps^{-2}n)$ time and is correct with high probability.
\end{restatable}

\begin{proof}
If $m \le \log^{1/\delta} n$, we run the algorithm of \cref{thm:small} with $s=n^{\delta/2}$  and $Q=[n]$, which runs 
	in time
	\[\Oh(\eps^{-2}(m^2 \log n + n))=\Oh(\eps^{-2} n).\]
	Otherwise, we run the algorithm of \cref{thm:main} with $s=n^{\delta/2}$ and $Q=[n]$, which runs in time
	\begin{multline*} \Oh\left(\sqrt{\tfrac{snm\log m}{\eps^5 k}}\log s + \tfrac{n\log s}{\eps^2 k}+ \tfrac{n}{\eps^2}\right)
		  = \Oh\left(\eps^{-2}\sqrt{n^{1+\delta/2}m^{1-\delta}}\log^{3/2} n + \tfrac{n \log n}{\eps m^{\delta}}+ \tfrac{n}{\eps^2} \right) \\
		  =  \Oh\left(\eps^{-2}n^{1-\delta/4}\log^{3/2} n + \tfrac{n}{\eps}+ \tfrac{n}{\eps^2}\right)  = \Oh(\eps^{-2} n).
	\end{multline*}
	The whole algorithm is then repeated $\Oh(1)$ times to lower the error probability.
\end{proof}

\section{Exact Algorithms}\label{sec:exact}

In this section, we focus on the following problem.

\defproblem{prob:exact}{Exact Text-To-Pattern Hamming Distances with a Fixed Threshold}{A text $T\in \Sigma^n$, a pattern $P\in \Sigma^m$, and a distance threshold $k$.}{For each position $i\in [n-m+1]$, compute the exact value $d_i = \HAM(P,T[i\dd i+m-1])$
	or state that $d_i > k$.}

Our approach is to first use \cref{thm:filter} for $\eps=\frac13$ in order to distinguish between positions $i$ with $\DD_i > \frac43k$ (which can be ignored due to $d_i>k$) and positions with $\DD_i \le \frac43 k$ (in which case $d_i \le 2k$ will be computed exactly). If there are few positions with $\DD_i \le \frac43k$,
then for each of them the \emph{kangaroo method} (LCE queries)~\cite{DBLP:journals/jcss/LandauV88} is used to determine $d_i$ in $\Oh(k)$ time
after $\Oh(n)$-time preprocessing. Otherwise, we prove that both the pattern $P$ and the parts of the text $T$ containing any approximate occurrence of $P$
are approximately periodic, i.e., that there is a value $\rho=\Oh(k)$ which is their $\Oh(k)$-period according to the following definition:

\begin{definition}
	An integer $\rho$ is a $d$-period of a string $X$ if $\HAM(X[0\dd x-\rho-1], X[\rho \dd x-1])\le d$.
\end{definition}

We first focus on the version of \cref{prob:exact} where $P$ and $T$ both have approximate period $\rho$ (which is also given as input). 
This version is studied in \cref{sec:periodic}, where we prove the following result.

\begin{restatable*}{theorem}{thmperiodic}\label{thm:periodic}
	Given an integer $\rho=\Oh(d)$, which is a $d$-period of both $P$ and $T$,
	\cref{prob:exact} can be solved in $\Oh(n+d \min(d,\sqrt{n\log n}))$ time and $\Oh(n)$ space using a randomized algorithm that returns correct answers with high probability.
\end{restatable*}

Combining \cref{thm:periodic} with the kangaroo method, we obtain the following result for the general case in \cref{sec:general}.
\begin{restatable*}{theorem}{thmexact}\label{thm:exact}
    There exists a randomized algorithm for \cref{prob:exact} that uses $\Oh(n)$ space, costs $\Oh(n+\min(\tfrac{nk^2}{m}, \tfrac{nk\sqrt{\log m}}{\sqrt{m}}))$ time,
	and returns correct answers with high probability.
\end{restatable*}

\subsection{The Case of Approximately Periodic Strings}\label{sec:periodic}

We start by recalling a connection,
originating from a classic paper by Fischer and Paterson~\cite{FisPat},
between text-to-pattern Hamming distances and the notion of a convolution of integer functions.
Throughout, we only consider functions $f : \Z \to \Z$ with finite support $\supp(f)=\{x : f(x)\ne 0\}$, that is, the number of non-zero entries in $f$ is finite.
The \emph{convolution} of two functions $f$ and $g$ is a function $f\ast g$ such that \[[f\ast g](i)\,=\,\sum_{j\in\Z} f(j)\cdot g(i-j).\]

For a string $X$ and a character $a\in \Sigma$,  the \emph{characteristic function} $X_a : \Z \to \{0,1\}$ is defined so that $X_a(i)=1$ if and only if $X[i]=a$.
The \emph{cross-correlation} of strings $X$ and $Y$ is a function $X\otimes Y$ defined as follows, with the reverse of $Y$ denoted by $Y^R$:
\[X\otimes Y\,=\,\sum_{a\in\Sigma} X_a\ast Y_a^R.\]

\begin{lemma}[{\cite{FisPat},~\cite[Fact 7.1]{CKP19}}]\label{lem:crosscorrelationHam}
For every $i\in [n-m+1]$, we have \[\HAM(P, T[i\dd i+m-1])=m-[T\otimes P](i+m-1).\]
\end{lemma}

Recall that in our setting $P$ and $T$ have a $d$-period $\rho=\Oh(d)$. 
For a function $f$  and an integer $\rho$, the \emph{forward difference} of $f$ with respect to $\rho$ is a function $\Delta_\rho[f]$ defined as $\Delta_\rho[f](i)= f(i+\rho)-f(i)$.
\begin{observation}[{\cite[Observation 7.2]{CKP19}}]\label{obs:diff}
	If $\rho$ is a $d$-period of a string $X$, then the characteristic functions $(X_a)_{a\in \Sigma}$ satisfy $\sum_{a\in \Sigma}|\supp(\Delta_\rho[X_a])| \le 2(d+\rho)$.
\end{observation}

In order to compute $T \otimes P$, one could sum up the convolutions $T_a \ast P_a^R$. 
However, the characteristic functions of $T$ and $P$ have total support size $\Theta(n+m)$, while  
the total support size of the forward differences of $T$ and $P$ with respect to $\rho$ is only $\Oh(d+\rho)$. Hence, it would be more efficient
to sum up the convolutions $\Delta_\rho[T_a] \ast \Delta_\rho[P_a^R]$ instead. This yields the \emph{second forward difference}
of $T\otimes P$ with respect to~$\rho$.

\begin{lemma}[see {\cite[Fact 7.4]{CKP19}}]\label{lem:second-diff-conv}
For strings $X,Y$ and a positive integer $\rho$, we have \[\Delta_\rho[\Delta_\rho[X \otimes Y]] \,=\, \sum_{a\in \Sigma} \Delta_\rho[X_a]\ast\Delta_\rho[Y_a^R].\]
\end{lemma}
Note that the second forward difference $\Delta_\rho[\Delta_\rho[f]]$, denoted by $\Delta^2_\rho[f]$, satisfies $\Delta^2_\rho[f](i)=f(i+2\rho)-2f(i+\rho)+f(i)$. Consequently, $T\otimes P$ can be retrieved using the following formula:
\[[T\otimes P](i)\,=\,\Delta^2_\rho[T \otimes P](i+2\rho)+2[T\otimes P](i+\rho)-[T\otimes P](i+2\rho).\]
Since $\supp(T\otimes P)\sub [n+m-1]$, it suffices to process subsequent indices $i$ starting from $i=n+m-2$ down to $i=0$.
Therefore, when computing $T\otimes P$, the values of $[T\otimes P](i+\rho)$ and $[T\otimes P](i+2\rho)$ have already been computed in previous iterations, and so the focus is on designing a mechanism for evaluating the function $\Delta^2_\rho[T \otimes P] = \sum_{a\in\Sigma}\Delta_\rho[T_a]\ast\Delta_\rho[P_a^R]$.

\paragraph{The convolution summation problem.}
In order to design a mechanism for evaluating $\Delta^2_\rho[T \otimes P]$, we introduce a more general \emph{convolution summation} problem which is stated as follows.
The input is two sequences of functions $\mathcal F=(f_1,f_2,\dots,f_\num)$ and $\mathcal G=(g_1,g_2,\dots g_\num)$,
and the output is the function $\mathcal F \otimes \mathcal G$ such that $[\mathcal F\otimes\mathcal G](i)=\sum_{j=1}^\num(f_j\ast g_j)(i)$.

We define the \emph{support} of  a sequence of functions $\mathcal H$ as $\supp(\mathcal H)=\bigcup_{h\in\mathcal H}\supp(h)$.
The total number of non-zero entries across  $h\in \mathcal H$ is denoted by
$\|\mathcal H\|= \sum_{h\in\mathcal H}|h|$, where $|h|=|\supp(h)|$.

In our setting, we assume that the input functions are given in an efficient \emph{sparse representation} (e.g., a linked list that contains only the non-zero entries). Moreover, the output of the algorithm is restricted to the non-zero values of $\mathcal{F\otimes G}$.

\begin{lemma}\label{lem:offline-convolution}
	There exists a randomized algorithm that, given two sequences of functions 
	$\mathcal{F}=(f_1,\ldots,f_\num)$ and $\mathcal{G}=(g_1,\ldots,g_\num)$ with non-empty supports such that
 $\supp(\F)\sub [n]$ and $\supp(\G)\sub [n]$,
	computes $\mathcal{F\otimes G}$ (correctly with high probability) in $O(n)$ space and in time \[O\left(\sum_{j=1}^t  \min(|f_j|  |g_j|,\,n\log n)\right)\:=\: O\left(\min\left(\|\F\|\|\G\|,\,(\|\F\|+\|\G\|)\sqrt{n\log n}\right)\right).\]
\end{lemma}
\begin{proof}
There are two methods that the algorithm chooses from to compute each convolution $f_j \ast g_j$.
The first method is to enumerate all pairs consisting of a non-zero entry in $f_j$ and in $g_j$.
Using standard hashing techniques, the time cost of computing the convolution $f_j \ast g_j$ this way is $O(|f_j| |g_j|)$.
The second method of computing $f_j \ast g_j$ is by FFT, which costs $O(n\log n)$ time.
The algorithm combines both methods by comparing $|f_j| |g_j|$ to $n\log n$ for each $1\le j \le t$ and picking the cheaper method for each particular $j$. Thus, the time for computing $f_j \ast g_j$ for any $j$ is $O(\min (|f_j| |g_j|,\,n\log n))$.

In order to reduce the space usage, the algorithm constructs $\mathcal{F\otimes G}$
by iteratively computing
the sum $\sum_{j=1}^i (f_j\ast g_j)$.
In each iteration, the algorithm adds the function  $f_j \ast g_j$ to the previously stored sum of functions.
The summation is stored using a lookup table of size $O(\min (n,\sum_{j=1}^\num |f_j||g_j|))$ via standard hashing techniques (notice that the exact size of the lookup table is pre-calculated).
The cost of adding $f_j \ast g_j$ to the previous sum of functions is linear in $\supp(f_j\ast g_j)$ and thus bounded by the time cost of computing $f_j \ast g_j$.
Hence, the total running time of the algorithm is $O\left(\sum_{j=1}^t \min(|f_j| |g_j|,\,n\log n)\right)$.

For each $j$, we have $|f_j|\le \|\F\|$, and therefore
\[
	\sum_{j=1}^t \min(|f_j|  |g_j|,n\log n) \le \sum_{j=1}^\num |f_j| |g_j|
\:\le\:\sum_{j=1}^\num \|\F\|  |g_j|
\:=\:\|\F\|\sum_{j=1}^\num   |g_j|
\:=\:\|\F\| \|\G\|.
\]
The second bound is obtained by recalling that $\min(x,y)\le \sqrt{xy} \le x+y$ holds for every positive $x$ and $y$:
\[\sum_{j=1}^t \min(|f_j| |g_j|,n\log n)
\le \sum_{j=1}^t \sqrt{|f_j|  |g_j|n\log n}
\le \sum_{j=1}^t (|f_j|+|g_j|)\sqrt{n\log n}
=(\|\F\|+\|\G\|)\sqrt{n\log n}.\qedhere
\]\end{proof}

\paragraph{The algorithm.}
We are now ready to describe and analyze the algorithm for the case of approximately periodic strings.

\thmperiodic
\begin{proof}
First, the algorithm constructs the forward differences $\Delta_{\rho}[P_a^R]$ and $\Delta_{\rho}[T_a]$. This step costs $\Oh(n)$ time.
Let $\F = (\Delta_{\rho}[T_a])_{a \in \Sigma}$ and $\G = (\Delta_{\rho}[P_a^R])_{a\in\Sigma}$.
The algorithm uses \cref{lem:offline-convolution} to compute $\F\otimes\G$.
Due to \cref{obs:diff}, $\|\F\|, \|G\|= \Oh(d)$, so this computation costs $\Oh(d\min(d, \sqrt{n \log n}))$ time and,
by \cref{lem:second-diff-conv},
results in $\Delta^2_\rho[T\otimes P]$ (in a sparse representation). Finally, the algorithm retrieves $T\otimes P$ and computes the Hamming distances using \cref{lem:crosscorrelationHam}. This final step costs $\Oh(n)$ time. Overall, the running time is $\Oh(n+d\min(d,\sqrt{n\log n}))$, and the space usage is $\Oh(n)$.
\end{proof}

\subsection{General Case}\label{sec:general}
\thmexact

\begin{proof}
Without loss of generality, we may assume that $k\ge \sqrt{m}$; otherwise, the stated running time is $\Oh(n)$ anyway.
Moreover, we assume that $n \le \frac32 m$; otherwise, the text $T$ can be decomposed into parts of length at most $\frac32 m$
with overlaps of length $m-1$, and each part of the text can be processed separately; the overall running time does not change since the running time for each part is linear in the length of the part.

First, the algorithm uses \cref{thm:filter} with $\eps = \frac13$, which results in a sequence $\tilde{d}_i$ satisfying the following two properties
with high probability: if $\DD_i > \frac43k$, then $d_i > k$; if $\DD_i \le \frac43k$, then $d_i \le 2k$.

 Let $C = \{i \in [n-m+1]: \tilde{d}_i \le \frac43 k\}$. Observe that we may assume without loss of generality that $\min C = 0$ and $\max C = n-m$;
 otherwise, $T$ can be replaced with $T[\min C \dd \max C + m-1]$ and all indices $i$ with $d_i \le k$ are preserved (up to a shift by $\min C$).

 We consider two cases depending on whether or not $C$ contains two distinct positions at distance $\rho \le \frac12k$ from each other.
 If $C$ does not contain two such positions, then $|C|=\Oh(\frac{n}{k})$, and the algorithm spends $\Oh(d_i)=\Oh(k)$ time for each $i\in C$
 to compute $d_i$ using $1+d_i$ \emph{Longest Common Extension} (LCE) queries. After $\Oh(n+m)$-time preprocessing, these queries locate in $\Oh(1)$ time the leftmost mismatch between any substrings of $T$ or $P$; see~\cite{DBLP:journals/jcss/LandauV88,DBLP:journals/jacm/Farach-ColtonFM00,DBLP:journals/jacm/KarkkainenSB06}. In the context of approximate pattern matching, this technique is known as the \emph{kangaroo method}; see~\cite{Ami}.
 In this case, the overall running time is  $\Oh(n)$.

 It remains to consider the case where $C$ contains two distinct positions at distance $\rho\le \frac12k$ from each other.
 We claim that in this case $\rho$ must be an $\Oh(k)$-period of both $P$ and $T$, and so applying \cref{thm:periodic} with $d=\Oh(k)$
 results in the desired running time and linear space usage.

 Let the positions at distance $\rho$ be $i$ and $i'$ with $i < i'=i+\rho$. Due to $\HAM(P, T[i\dd i+m-1])\le 2k$ and $\HAM(P, T[i'\dd i'+m-1])\le 2k$,
 we conclude from the triangle inequality that:
%
 \begin{align*}
 \HAM(P[0\dd m-\rho-1],P[\rho\dd m-1]) & \le  \HAM(P[0\dd m-\rho-1], T[i'\dd i'+m-\rho-1])\\
 & \qquad \qquad \qquad \quad+\HAM(T[i+\rho\dd i+m-1],P)\\
 & \le 2k+2k=4k.
 \end{align*}
 Hence, $\rho$ is a $4k$-period of $P$. Furthermore, due to $\HAM(P, T[0\dd m-1])\le 2k$ (since $0\in C$), $\rho$ is an $8k$-period of $T[0\dd m-1]$.
 Similarly, $\rho$ is an $8k$-period of $T[n-m\dd n-1]$ (since $n-m\in C$). As $n \le \frac32 m \le 2m-\rho$,
 these two fragments of $T$ overlap by at least $\rho$ characters, which implies that $\rho$ is a $16k$-period of $T$.
 This completes the proof.
\end{proof}

Note that if one is interested in just an $O(n + \frac{nk^2}{m})$ upper
bound (which is sufficient for the application in the next section), then the algorithm does not need FFT (as the weaker $O(\|\F\|\|\G\|)$ upper bound in \cref{lem:offline-convolution} suffices).

\section{Combining Algorithms}\label{sec:combo}

In this section we return to approximation algorithms and design an almost linear time solution for \cref{prob:apx}, and a linear time solution for \cref{prob:fixedk}
by combining the three algorithms from \cref{sec:offline,sec:exact,sec:small}.

\begin{corollary}\label{cor:combo}
There exists a randomized algorithm for \cref{prob:fixedk} that runs in
$O(\eps^{-2} n)$ time and is correct with high probability.
Moreover, there exists a randomized algorithm for \cref{prob:apx} that runs
in $O(\eps^{-2} n\log\log n)$ time and is correct with high probability.
\end{corollary}
\begin{proof}
We consider three cases.
\begin{itemize}
\item
Case I: $m\le \log^{2} n$.  We run the algorithm of \cref{thm:small} to solve \cref{prob:fixedk} in
$O(\eps^{-2} n)$ time.  
We solve \cref{prob:apx}
by
examining all $k \le m$ that are powers of 2,
in $O(\eps^{-2}  n\log m)= O(\eps^{-2}  n\log\log n)$ time (see the discussion in \cref{sec:preliminaries}).

\item
Case II: distances $d_i \le \eps^{-1} \sqrt{m}$.  We run the exact algorithm of \cref{thm:exact},
which computes all such distances in $O(n+(\eps^{-1}\sqrt{m})^2\frac{n}{m})=\Oh(\eps^{-2}n)$ time.
\item

Case III:  distances $d_i > \eps^{-1} \sqrt{m}$ and $m > \log^{2}n$.  We run the
algorithm of \cref{thm:main} with $s=n^{0.25}$
to solve \cref{prob:fixedk}
in time \[O\left(\sqrt{\tfrac{s nm \log m}{\eps^5 k}}\log s + \tfrac{n \log s}{\eps^2 k} + \tfrac{|Q|}{\eps^2}\right)
= O\left(\eps^{-2}\sqrt{n^{1.25}m^{0.5}}\log^{1.5}n + \tfrac{n \log n}{\eps\sqrt{m}} + \tfrac{|Q|}{\eps^2}\right)=O(\eps^{-2}n).\]

We solve \cref{prob:apx} by examining all $k>\eps^{-1} \sqrt{m}$ that are powers of 2 (in parallel) and performing a binary search
for each $i\in Q$.  The total time is
\begin{multline*}
O\left(\sum_{\substack{k>\eps^{-1}\sqrt m \\ k \text{ is a power of } 2}} \sqrt{\tfrac{s nm \log m}{\eps^5 k}}\log s + \sum_{\substack{k>\eps^{-1}\sqrt m \\ k \text{ is a power of } 2}} \tfrac{n\log s}{\eps^2 k} + \tfrac{|Q|\log \log m}{\eps^2}\right) 
\\  = O\left(\eps^{-2}\sqrt{n^{1.25}m^{0.5}}\log^{1.5}n + \tfrac{n\log n}{\eps\sqrt{m}} + \tfrac{|Q|\log \log m}{\eps^2}\right)
  =  O(\eps^{-2} n \log \log m).
\end{multline*}
The algorithm is repeated $\Oh(1)$ times to lower the error probability.\qedhere
\end{itemize}
\end{proof}

In \cref{app:lin},
we describe further improvements to \cref{cor:combo}, reducing the running time to linear
(and even slightly sublinear), by using more complicated bit-packing tricks.

\section{Algorithms with Improved $\eps$-Dependence}\label{sec:mm}

In this section, we show that \cref{prob:apx} can be solved in $\Ohtilde(n)$ time without any $\eps^{-O(1)}$ factors when the pattern is sufficiently long, namely, when $m > \eps^{-27.22}$.
For this, we combine our generic sampling algorithm of \cref{sec:generic} with rectangular matrix multiplication~\cite{DBLP:journals/siamcomp/Coppersmith82,DBLP:conf/soda/GallU18}. Specifically, we show that if an $n \times n^{\alpha}$ matrix and an $n^{\alpha} \times n$ matrix can be multiplied in $\Ohtilde(n^2)$ time,
then \cref{prob:apx} can be solved in $\Ohtilde(n)$ time if $m  > \eps^{-\max\big(4+\tfrac{4}{\alpha}, 10\big)}$. In particular, with $\alpha > 0.17227$ due to Coppersmith~\cite{DBLP:journals/siamcomp/Coppersmith82},
the constraint reduces to $m >  \eps^{-27.22}$. Allowing $\OO(n)$ time rather than $\Ohtilde(n)$ time, we can use a more recent result by Le Gall and Urrutia~\cite{DBLP:conf/soda/GallU18} with $\alpha > 0.3138$, resulting 
in a looser constraint $m >  \eps^{-16.75}$.
We would like to remark, though, that in this version of the manuscript, these exponents 27.22 and 16.75
have not been optimized.

We start with a solution to \cref{prob:fixedk}.
\begin{theorem}\label{thm:fixedmm}
If $k > \eps^{-\max\big(2+\tfrac{2}{\alpha},5\big)}$ and $n > \eps^{-\max\big(\tfrac{4}{\alpha}, 6\big)}$, then \cref{prob:fixedk} can be solved in $\Ohtilde(n)$ time using a randomized algorithm returning correct answers with high probability.
\end{theorem}
\begin{proof}
We apply the approach of \cref{sec:generic} with $z = \min(\eps^2 k, \sqrt{n})$ and a sufficiently large $s = \Oh(1)$. As in the proof of \cref{thm:main},
we map the strings $X_{u}^{(\ell)}$ and $Y_{v}^{(\ell)}(i)$ to $x_u^{(\ell)},y_v^{(\ell)}(i)\in \{0,1\}$ using Karp--Rabin fingerprints composed with random hash functions.
Let $\vec{x}_u=\langle x_{u}^{(\ell)} : \ell \in [L]\rangle $ and $\vec{y}_v(i)=\langle y_{v}^{(\ell)}(i) : \ell \in [L]\rangle $ be the vectors defined in the proof of \cref{thm:main}; here, we do not pack these bit vectors, though.
Recall that the vectors $\vec{x}_u$ for $u\in [z]$ can be constructed in time $\Oh(\beta m z L)=\Oh(\frac{m z}{\eps^2 k})=\Oh(m)$.
Similarly, the vectors $\vec{y}_v(i)$ for $v\in [\ceil{p/z}]$ can be maintained (for subsequent $i\in [n-m+1]$) in the overall time $\Oh(\beta n L p / z)= \Ohtilde(\frac{n}{\eps^3 z})=\Ohtilde(\frac{n}{\eps^5 k}+\frac{\sqrt{n}}{\eps^3})$. Since $k > \eps^{-5}$ and $n >\eps^{-6}$, this time is $\Ohtilde(n)$.

It remains to implement Lines~\ref{ln:query-generic}--\ref{ln:output} of \cref{alg:main}. For each $i$, a naive implementation costs $\Ohtilde(\eps^{-2})$ time,
where the bottleneck is computing $c_i$, which is the inner product of $\vec{x}_{u_i}$ with $\vec{y}_{v_i}(i)$; the remaining operations cost $\Oh(1)$ time for each $i$.
We speed up these computations by arranging \emph{distinct} vectors $\vec{x}_{u_i}$ and  $\vec{y}_{v_i}(i)$ into two matrices and multiplying the two matrices.

The number of distinct vectors $\vec{x}_{u_i}$ is at most $z$. The analysis for vectors $\vec{y}_{v_i}(i)$ is more involved:
First, note that $v_i$ changes $\Oh(n/z)$ times as $i$ increases from $0$ to $n-m$. Secondly, observe that $\vec{y}_{v}(i)$ differs from $\vec{y}_{v}(i-1)$
at a given coordinate $\ell$ with probability $\Oh(\beta)=\Oh(1/k)$. Applying a union bound, $\Pr[\vec{y}_{v}(i)\ne \vec{y}_{v}(i-1)] = \Oh(\frac{1}{\eps^2 k})$. Hence, the expected number of distinct vectors $\vec{y}_{v}(i)$ is $\Oh(\frac{n}{z}+\frac{n}{\eps^2 k})$.
The algorithm declares a failure if this quantity exceeds the expectation by a large constant factor (the constant probability of this event adds up to the constant probability of the algorithm returning incorrect answers). 
Consequently, our task reduces to multiplying two matrices of dimensions $\Oh(z)\times \Oh(\eps^{-2})$ and $\Oh(\eps^{-2})\times \Oh(\frac{n}{z}+\frac{n}{\eps^2 k})$.  Since $\frac{n}{z} \ge z$, this process takes $\Ohtilde(n + \frac{nz}{\eps^2 k})=\Ohtilde(n + \frac{n\eps^2 k}{\eps^2 k})=\Ohtilde(n)$
time provided that $z^{\alpha} > \eps^{-2}$, which follows from $z^{\alpha} = \eps^{2\alpha} k^{\alpha} > \eps^{2\alpha-2\alpha-\tfrac{2\alpha}{\alpha}}=\eps^{-2}$
or $z^{\alpha} =n^{\tfrac{\alpha}{2}} > \eps^{\tfrac{-4}{\alpha}\cdot \tfrac{\alpha}{2}}=\eps^{-2}$.

This way, we obtained an algorithm with expected running time $\Ohtilde(n)$ and with small constant probability of error for every position $i\in Q$.
We repeat the algorithm $\Oh(\log n)=\Ohtilde(1)$ times to achieve with high probability bounds on both correctness and running time.
\end{proof}

\begin{corollary}\label{thm:apxmm}
If $m > \eps^{-\max\big(4+\tfrac{4}{\alpha}, 10\big)}$, then \cref{prob:apx} can be solved in $\Ohtilde(n)$ time using a randomized algorithm returning correct answers with high probability.
\end{corollary}
\begin{proof}
	We apply an exact $\Ohtilde(n)$-time algorithm~\cite{Cli} for $k = \sqrt{m}$ (see also \cref{thm:exact}) to determine $d_i$ at locations $i$ for which $d_i \le \sqrt{m}$.
	As for the distances $d_i \ge \sqrt{m}$, we apply the algorithm in \cref{thm:fixedmm} for all $\tfrac12 \sqrt{m}\le k \le m$ that are powers of two.
	In this setting, we have $k > \sqrt{m} > \eps^{-\max\big(2+\tfrac{2}{\alpha},5\big)}$
	and $n \ge m  > \eps^{-\max\big(4+\tfrac{4}{\alpha}, 10\big)}> \eps^{-\max\big(\tfrac{4}{\alpha}, 6\big)}$, so the running time of each call is $\Ohtilde(n)$, and the number of calls is $\Oh(\log m)=\Ohtilde(1)$.
\end{proof}







\section{Sublinear-Time Algorithms}\label{sec:sublin}

\newcommand{\xx}{\vec{x}}
\newcommand{\yy}{\vec{y}}

In this section, we show how to find locations with
Hamming distance approximately (up to a constant factor) less than  a fixed threshold value $k$
in truly sublinear time, provided that $k$ is not too small and the number of occurrences to report is sublinear. 
In comparison, our earlier running times have an $\Omega(|Q|)$ term, which is at least linear in the worst case.


We use known data structures for high-dimensional approximate spherical range reporting (which is related to approximate nearest neighbor search):

\begin{lemma}\label{lem:ann}
Given a constant $c>1$, let $\rho_q$ and $\rho_u$ be parameters satisfying
$c\sqrt{\rho_q} + (c-1)\sqrt{\rho_u} = \sqrt{2c-1}$.

Let $\xx_1,\ldots,\xx_n$ be vectors in $\{0,1\}^d$, and let $k\in [d]$.
In $\OO(dn^{1+\rho_u})$ time one can build a data structure that supports the following operations:
\begin{enumerate}[label={(\roman*)}]
\item given any query vector $\yy\in\{0,1\}^d$, report a
set $A$ satisfying $\{i : \HAM(\xx_i,\yy)\le k\}\subseteq A\subseteq \{i : \HAM(\xx_i,\yy)\le ck\}$
with high correctness probability, in $\OO(d(n^{\rho_q} + |A|n^{\rho_u}))$ time;
\item given any query vector $\yy\in\{0,1\}^d$ and query interval $I$, report a set $A$ satisfying $\{i\in I: \HAM(\xx_i,\yy)\le k\}\subseteq A\subseteq \{i\in I: \HAM(\xx_i,\yy)\le ck\}$
with high correctness probability, in $\OO(d(n^{\rho_q} + |A|n^{\rho_u}))$ time.
\end{enumerate}
\end{lemma}
\begin{proof}
Andoni et al.~\cite{AndoniLRW17} gave (randomized) dynamic data structures
for $c$-approximate nearest neighbor search in Hamming space, using data-dependent locality-sensitive hashing: with
the time $\OO(dn^{\rho_u})$ per update (insertion or deletion), the query time
is $\OO(dn^{\rho_q})$ for parameters $\rho_u$ and $\rho_q$ satisfying the stated equation.
Ahle et al.~\cite[Appendix~E]{AhleAP17} observed that such a data structure can be used to answer $c$-approximate spherical range reporting queries in
$\OO(d(n^{\rho_q} + |A|n^{\rho_u}))$ time.  This proves part~(i).

Part (ii) follows from part (i) by a standard technique (namely, one-dimensional range trees): for each dyadic\footnote{
A \emph{dyadic} interval is an interval of the form $[2^ij, 2^i(j+1))$ for integers $i,j$.
} 
interval $J$, we build the data structure from part (i) for the subset $\{\xx_i: i\in J\}$.
The preprocessing time and space increase only by a logarithmic factor.
A query interval $I$ can be decomposed into a union of $O(\log n)$ disjoint dyadic intervals.  So, the query time also increases only by a logarithmic factor.
\end{proof}

\begin{theorem}\label{thm:sublin}
Given a constant $c>1$, let $\rho_q$ and $\rho_u$ be parameters satisfying
$c\sqrt{\rho_q} + (c-1)\sqrt{\rho_u} = \sqrt{2c-1}$.

Let $\eps\in(0,1)$ be an arbitrarily small constant.
Given a text $T\in \Sigma^n$, a pattern $P\in \Sigma^m$, and an integer $k\le m$,
there is a randomized algorithm to report a set of locations, such that
every location of the text with Hamming distance at most $(1-\eps)k$ is reported, and every reported location has Hamming distance at most $(1+\eps)ck$.  The algorithm is correct with high probability and has expected running time
\[\OO_\eps\left(n^{\frac{1+\rho_u}{2+\rho_u-\rho_q}} +
\frac{n}{(\tfrac{kn}{m})^{\frac{1-\rho_q}{3-3\rho_q+\rho_u}}} + \tfrac{n}{k^{1-\rho_q}} +
\OCC\cdot \min\left\{ n^{\frac{\rho_u}{2+\rho_u-\rho_q}},\,
(\tfrac{kn}{m})^{\frac{\rho_u}{3-3\rho_q+\rho_u}},\,
k^{\rho_u}\right\}\right),\]
where $\OCC$ is the number of locations in the text with Hamming distance at most $(1+\eps)ck$.
\end{theorem}
\begin{proof}
We follow our generic algorithm but with a few modifications.
We reset $\beta = \frac{\tteps}{ck}$ for some constant $\tteps=\Theta(\eps)$.
In \cref{ln:constructX,ln:constructY}, we replace the strings $X_u^{(\ell)}$ and
$Y_v^{(\ell)}(i)$ with their hashed fingerprints: $x_u^{(\ell)}=h^{(\ell)}(F^{(\ell)}(X_u^{(\ell)}))$ and
$y_v^{(\ell)}(i)=h^{(\ell)}(F^{(\ell)}(Y_v^{(\ell)}(i)))$, for
randomly chosen functions $F^{(\ell)}\in {\cal F}_M$ and
$h^{(\ell)}\in {\cal H}_M$, where $M=n^{\Oh(1)}$ is a sufficiently large prime.
We do not explicitly store $y_v^{(\ell)}(i)$
for all~$i$, but just for those $i$ for which $Y_v^{(\ell)}(i)$ changes as $i$ increases.
As before, Lines~\ref{ln:pick}--\ref{ln:constructY} take
$\OO_\eps(\sqrt{\frac{snm}{k}} + \frac{n}{k})$ time.
However, to aim for sublinear total time, we need to implement
Lines~\ref{ln:query-generic}--\ref{ln:output} differently.

Recall that $u_i\in [z]$ and $v_i\in[\ceil{p/z}]$ are indices
defined to satisfy $(i\bmod p) = u_i + v_iz$.
As $i$ increases, if the value $y_{v_i}^{(\ell)}(i)$ changes for some $\ell\in [L]$,
we say that the index $i$ is \emph{critical}.
We reuse an argument from the proof of \cref{thm:fixedmm} to bound the number of critical indices:
The index $v_i$ changes $O(n/z)$ times.
If $v_i$ is unchanged as $i$ increments,
then $y_{v_i}^{(\ell)}(i)$ changes
only when $(i-vz)\bmod p$
or $(i+m-vz)\bmod p$ is in $B^{(\ell)}$, which
happens with probability $O(\beta)=O_\eps(1/k)$.
Thus, the expected number of critical indices is $O_\eps(\frac nz + \frac nk)$ for each fixed $\ell$, and remains $\tO_\eps(\frac nz + \frac nk)$ over all logarithmically many $\ell\in [L]$.

We build the data structure from \cref{lem:ann} storing the vectors
$\xx_u=\langle x_u^{(\ell)}:\ell\in [L]\rangle$ for all $u\in [z]$,
in $\OO(z^{1+\rho_u})$ time.
Consider two consecutive critical indices $a$ and $b$.
For all $i\in [a,b)$, the vector $\yy_{v_i}=\langle y_{v_i}^{(\ell)}:\ell\in [L]\rangle$ is unchanged.
We report an index set, where every index $i\in [a,b)$ such that $\xx_{u_i}$ has Hamming distance
at most $\tfrac12 (1-(1-\beta)^k) L$ from $\yy_{v_i}$ is reported,
and every reported index $i\in [a,b)$ has Hamming distance
at most $\tfrac12 (1-(1-\beta)^{(1+\teps)ck}) L$,
for some appropriate choice of $\teps=\Theta(\eps)$.
This reduces to the type of query supported by \cref{lem:ann}(ii),
since $i\mapsto u_i$ maps $[a,b)$ into at most two intervals.
Note that
the ratio
$\frac{1-(1-\beta)^{(1+\teps)ck}}{1-(1-\beta)^k}
\ge \frac{(1+\teps)ck\beta - O(ck\beta)^2}{k\beta}
\ge (1+\teps)c - O(\tteps c^2)$
exceeds $c$, by choosing $\tteps=\Theta(\teps)$ with a sufficiently small constant factor.

By a similar probabilistic analysis as before, the error probability per $i$ is $O(1/s)$.  Thus, the total expected number of indices reported is $O(\OCC + \frac ns)$.  The time to answer all $\tO_\eps(\frac nz + \frac nk)$ queries using \cref{lem:ann}(ii) is
$\OO_\eps((\frac nz + \frac nk) z^{\rho_q} + (\OCC + \frac ns)z^{\rho_u})$ in expectation.
The overall expected running time is
$$\OO_\eps\left(\sqrt{\tfrac{snm}{k}} + z^{1+\rho_u} +
\left(\tfrac{n}{z} + \tfrac nk\right) z^{\rho_q} + \left(\OCC + \tfrac ns\right)z^{\rho_u}\right).$$

To balance all the terms,
set $s=z^{2\rho_u/3}(\frac{kn}{m})^{1/3}$
and $z=\min\{n^{1/(2+\rho_u-\rho_q)}, \, (\frac{kn}{m})^{1/(3-3\rho_q+\rho_u)},
\, k\}$.
Then the expected running time is bounded by the expression stated in the theorem.

The error probability per location is smaller than a constant $<\tfrac 12$.  We can lower the error probability by repeating logarithmically many times and outputting a location when it lies in a majority of all the reported sets.
\end{proof}

\begin{example}
For $\rho_u=\rho_q=1/(2c-1)$ and for all $m\le n$, the time bound is at most
$$\OO_\eps\left( n^{c/(2c-1)} + \tfrac{n}{k^{(2c-2)/(6c-5)}}
 + \OCC\cdot n^{1/(4c-2)}\right).$$
For $\rho_u=0$ and $\rho_q = (2c-1)/c^2$,
the time bound is at most
$$\OO_\eps\left(n^{c^2/(2c^2-2c+1)} +
\tfrac{n}{k^{1/3}} + \tfrac{n}{k^{(c-1)^2/c^2}} +
\OCC\right).$$
In particular, in the case of $c=2$, the above bounds are
$\OO_\eps(n^{2/3} + n/k^{2/7} + \OCC\cdot n^{1/6})$
and $\OO_\eps(n^{4/5} + n/k^{1/4} + \OCC)$, though other tradeoffs are possible.
\end{example}

For $c$ sufficiently close to~1, one can do better by using known \emph{offline} approximate nearest neighbor algorithms:

\begin{lemma}\label{lem:ann2}
Let $c=1+\eps$ for a sufficiently small constant $\eps>0$.
A batch of $n$ offline queries of the type in \cref{lem:ann}(i) and (ii) can be answered in
\[O\left(d^{O(1)}\left(n^{2-\Omega(\eps^{1/3}/\log^{2/3}(1/\eps))} + {\cal A} n^{O(\eps^{1/3}/\log^{2/3}(1/\eps))}\right)\right)\] time with high correctness probability,
where ${\cal A}$ is the total size of the reported sets.
\end{lemma}
\begin{proof}
Alman, Chan, and Williams~\cite{AlmanCW16,AlmanCW20}
gave randomized algorithms for
offline $(1+\eps)$-approximate nearest neighbor search,
via the polynomial method and rectangular matrix multiplication:
the running time for $n$ queries is
$O\big(d^{O(1)}n^{2-\Omega(\eps^{1/3}/\log^{2/3}(1/\eps))}\big)$.
It is straightforward to modify their algorithms for
$(1+\eps)$-approximate spherical range reporting in the time bound
stated in the lemma.  This proves part (i).

Part (ii) follows from part (i) by the same argument as before
using dyadic intervals (which carries over to the offline setting).
\end{proof}

\begin{theorem}
Let $c=1+\eps$ for a sufficiently small constant $\eps>0$.
The expected running time in \cref{thm:sublin} is at most
\[ \tO\left( \frac{n}{k^{\Omega(\eps^{1/3}/\log^{2/3}(1/\eps))}} +
\OCC\cdot k^{O(\eps^{1/3}/\log^{2/3}(1/\eps))} \right).\]
\end{theorem}
\begin{proof}
We proceed as in the proof of \cref{thm:sublin}, but note that
the generated queries are offline, for which \cref{lem:ann2}
is applicable.  Effectively, we can set $\rho_q=1-\Theta(\eps^{1/3}/\log^{2/3}(1/\eps))$ and $\rho_u=\Theta(\eps^{1/3}/\log^{2/3}(1/\eps))$ in the time bound.
\end{proof}

In the case of distinguishing between distance 0 (exact match) versus distance more than $\delta m$, the algorithm in \cref{thm:sublin} can be simplified:

\begin{theorem}\label{thm:prop:test}
Given a text string of length $n$, a pattern string of length $m$, and a value $\delta > 0$,
there is a randomized algorithm to report a set of locations, such that
every location of the text with Hamming distance $0$ is reported, and every reported location has Hamming distance at most $\delta m$.  The algorithm is correct with high probability and has expected running time
\[\tO\left(\delta^{-1/3}n^{2/3} + \delta^{-1} \tfrac{n}{m} + \OCC\right),\]
where $\OCC$ is the number of locations in the text with Hamming distance at most $\delta m$.
\end{theorem}
\begin{proof}
We proceed as in the proof of \cref{thm:sublin}, using specific constants for $c$ and $\eps$ (e.g., $c=2$ and $\eps=1/3$) and setting $k=\frac{\delta m}{(1+\eps)c}=\Theta(\delta m)$.

We no longer need \cref{lem:ann} (approximate nearest neighbor
search).  We can use standard hashing and one-dimensional range search to find all $i\in [a,b)$ such that
$\xx_{u_i}$ has distance 0 from (i.e., is identical to) $\yy_{v_i}$, for each pair of  consecutive critical indices $a$ and $b$.  (The probabilistic analysis can also be simplified, with no Chernoff bounds needed.)  Effectively, we can set $\rho_q=\rho_u=0$ in the time bound, which becomes
$\tO(\sqrt{n} + \frac{n^{2/3}m^{1/3}}{k^{1/3}} + \frac nk +\OCC)$.  Putting $k=\Theta(\delta m)$ gives the theorem.
\end{proof}

The above implies a sublinear-time \emph{property tester\/} for pattern matching:
run the algorithm for $\tO\left(\delta^{-1/3}n^{2/3} + \delta^{-1} \frac{n}{m}\right)$
steps, and return ``true'' if the algorithm has not run to completion
or at least one location has been reported.
This way, if an exact match exists, then ``true'' is returned with high probability; and if the pattern is $\delta$-far (i.e., has Hamming distance more than $\delta m$) from the text at every location, then ``false'' is returned with probability at least a constant $>\frac 12$ (which can be amplified by repetition).

\begin{remark}
There has  been some past work on sublinear-time algorithms for string problems.  Chang and Lawler~\cite{ChangL94} considered the exact fixed-threshold problem and described an algorithm with expected time
$O(\frac{kn}{m}\log_\sigma m)$, which is sublinear when $k$ is small (and $m$ is not too small), but their work assumes a uniformly random text string.
Andoni et al.~\cite{AndoniIKH13} gave a sublinear-time algorithm for a \emph{shift-finding} problem that is closely related to the approximate $k$-mismatch problem (their algorithm similarly uses approximate nearest neighbor search as a subroutine), but
their work assumes that the pattern string is uniformly random and the
text is generated by adding random (Gaussian) noise to a shifted copy of the pattern.
By contrast, our results hold for \emph{worst-case} inputs.
Truly sublinear-time algorithms have been proposed for the problem
of approximating the edit distance between two strings, by
Batu et al.~\cite{BatuEKMRRS03} and Bar-Yossef et al.~\cite{Bar-YossefJKK04}, but with large (polynomial) approximation factors.
Bar-Yossef et al.~\cite{Bar-YossefJKK_APPROX04} studied the ``sketching
complexity'' of pattern matching and obtained sublinear bounds of the
form $\tO(\delta^{-1}\frac nm)$, but these do not correspond to actual
running times.

\end{remark}



\section{Streaming Algorithms}\label{sec:stream}

We now consider approximation algorithms in the streaming model.

\paragraph{Multi-stream dictionary matching.}
A useful building block for our algorithm is a subroutine for the multi-stream dictionary matching problem.
A \emph{dictionary} $D$ is a set of patterns of length at most $m$ each. 
In addition, there exist several streams representing different texts, and at each time step a new character arrives in one of the streams.
After the arrival of a character to the $i$th stream, the algorithm has to report the longest pattern from $D$ that matches a suffix of the $i$th text stream, or state that none of the patterns from $D$ is a suffix of the $i$th text stream.
We use the algorithm of Golan et al.~\cite{GKP18} for the multi-stream dictionary problem:
\begin{lemma}[{immediate from~\cite[{Theorem 2}]{GKP18}}]\label{lem:multistream_dictionary}
	There exists an algorithm for the multi-stream dictionary matching problem on a dictionary $D$, with $d$ patterns of length at most $m$ each, which for $t$ text streams costs $O(d\log m +t\log m \log d)$ words of space and $O(\log m+\log d\log\log d)$ time per character. Both these complexities are worst-case, and the algorithm is correct with high probability.
\end{lemma}

\subsection{Algorithm for \cref{prob:fixedk}}
In this section, we prove the following theorem.
\begin{theorem}\label{thm:streaming-prob2}
	There exists a streaming algorithm for \cref{prob:fixedk} where the pattern $P$ can be preprocessed in advance and the text arrives in a stream
	so that $\tilde d_{i-m+1}$ is reported as soon as $T[i]$ arrives.
	The space usage of the algorithm is $\tO(\min(\eps^{-2.5}\sqrt {k},\eps^{-2}\sqrt m))$ words, the running time per character is $\tO(\frac {\min(\eps^{-2.5}\sqrt k,\eps^{-2}\sqrt m)}{k}+\eps^{-2})$, and the outputs are correct with high probability.
\end{theorem}

We implement \cref{alg:main} in the streaming model with $z=\down{\sqrt p}$ (and therefore $\ceil{\frac pz}=\Theta(\sqrt{p})$) and $s=\Theta(1)$.
Recall that $L=\eps^{-2}\log s$ and $\hat p = \eps^{-1} s k \log m$.
For each $\ell\in[L]$, let $D^{(\ell)}=\{X_u^{(\ell)} : u\in[z]\}$ be a dictionary with $z$ strings.
For each $\ell\in[L]$ and $v\in [\ceil{p/z}]$, let $Y_v^{(\ell)}$ be a stream such that at time $i$ (after the arrival of $T[i]$)  \[Y_v^{(\ell)} \ = \bigodot_{j\le i\;:\; (j-vz) \bmod p \;\in\; B^{(\ell)} } T[j].\] Notice that after the arrival of $T[i]$, we have that $Y_v^{(\ell)}(i-m+1)$ is a suffix of $Y_v^{(\ell)}$.

\paragraph{Preprocessing phase.}
During the preprocessing phase, the algorithm chooses a random prime $p\in [\hat p , 2\hat p)$, and for each $\ell\in[L]$  the algorithm picks a random  sample $B^{(\ell)}\subseteq [p]$ with sampling rate $\beta=\frac 1{2k}$. For each $\ell \in [L]$, the algorithm (separately) applies the preprocessing of the multi-stream dictionary algorithm of \cref{lem:multistream_dictionary} on each $D^{(\ell)}$ so that the patterns from $D^{(\ell)}$ can be matched against the streams $Y_v^{(\ell)}$ for $v\in [\ceil {p/z}]$.

\paragraph{Processing phase.}
After the arrival of $T[i]$, the algorithm appends $T[i]$ into some of the streams $Y_v^{(\ell)}$.
More precisely, $T[i]$ should be inserted into $Y_v^{(\ell)}$ if and only if $(i-vz)\bmod p\in B^{(\ell)}$. 
A stream $Y_v^{(\ell)}$ is called \emph{active at time $i$} if and only if $(i-vz)\bmod p\in B^{(\ell)}$.
The following lemma is useful for efficiently retrieving the active streams at any time.

\begin{lemma}\label{lem:streams-scheduler}
There exists a data structure that at any time $i$ reports the streams active at time $i$. 
The space usage of the data structure is linear in the total number of streams, and the query time is linear in the output size (the number of active streams at time $i$) with high probability.
\end{lemma}

\begin{proof}
The data structure maintains one handle for each stream.
These handled are stored in a hash table that maps future time-points into linked lists of streams' handles. 
The algorithm preserves an invariant that at any time $i$, the handle of any stream $Y_v^{(\ell)}$ is stored in the linked list of the smallest $j\ge i$ such that $Y_v^{(\ell)}$ is active at time $j$ (i.e. $(j-vz)\bmod p\in B^{(\ell)}$).
The algorithm maintains each set $B^{(\ell)}$ as a cyclic linked list, and each stream handle $Y_v^{(\ell)}$ is maintained in the linked list of time $j$ with a pointer to the element $(j-vz)\bmod p$ in the cyclic linked list of $B^{(\ell)}$.

During an update (incrementing $i$ to $i+1$), the algorithm first reports all the streams in the linked list of time $i$.
Then, in order to keep the data-structure up-to-date and preserve the invariant, the algorithm computes for each stream $Y_v^{(\ell)}$ that is active at time $i$ the smallest $j>i$ such that $Y_v^{(\ell)}$  is active also at time $j$. 
This computation is done in constant time per stream by advancing the pointer to the cyclic linked list of $B^{(\ell)}$.
Then, the algorithm inserts the handle of $Y_v^{(\ell)}$ into the linked list of time $j$.
Finally, the algorithm removes the empty linked list of time $i$ from the hash table to reduce the space usage.
\end{proof}

Using the data structure of \cref{lem:streams-scheduler}, the algorithm retrieves all the active streams and passes $T[i]$ into each of those streams.
After processing $T[i]$, the dictionary matching algorithm of \cref{lem:multistream_dictionary} identifies for each stream $Y_v^{(\ell)}$ the current longest suffix that matches a pattern in $D^{(\ell)}$.
The algorithm maintains a pointer $\pi_v^{(\ell)}$ to the longest pattern from $D^{(\ell)}$ that is a current suffix of $Y_v^{(\ell)}$, if such a pattern exists. Maintaining these pointers costs constant space per stream, and the overall time cost per text character is linear in the number of currently active streams. 

\paragraph{Evaluating $\tilde d_{i-m+1}$.}
After updating all of the active streams the algorithm estimates $d_{i-m+1}$ by applying Lines~\ref{ln:query-generic}--\ref{ln:output} from~\cref{alg:main}.
In order to test whether $X_{u_{i-m+1}}^{(\ell)} = Y_{v_{i-m+1}}^{(\ell)}(i-m+1)$ the algorithm checks if $X_{u_{i-m+1}}^{(\ell)}$ is a suffix of the pattern pointed to by $\pi_{v_{i-m+1}}^{(\ell)}$.\footnote{The test costs constant time using standard techniques.}

\paragraph{Complexity analysis.}
For each $\ell\in[L]$, the dictionary $D^{(\ell)}$ contains $O(\sqrt p)$ patterns of length $O(m)$. Moreover, the number of streams $Y_v^{(\ell)}$ is also $O(\sqrt p)$. Thus, the algorithm of \cref{lem:multistream_dictionary} uses $O(\sqrt p\log m+\sqrt p\log m\log d)=\tilde O(\sqrt p)=\tilde O(\min(\sqrt{\eps^{-1} k},\sqrt m))$ space.
Summing over all $\ell\in [L]$, the space usage of all the streams is $O(|L|\cdot \sqrt{p})=\tilde O(\min(\eps^{-2.5}\sqrt {k},\eps^{-2}\sqrt m))$ words of space. Since the auxiliary data structure of \cref{lem:streams-scheduler} takes linear space in the number of streams, the total space usage of the algorithm is $\tO(\min(\eps^{-2.5}\sqrt {k},\eps^{-2}\sqrt m))$.

As for the running time, we first bound the number of active streams at any time. For any stream $Y_v^{(\ell)}$ and any time $i$, the stream is active at time $i$ if and only if $(i-vz)\bmod p\in B^{(\ell)}$, which happens with probability $\beta$ independently across all the streams. 
By standard Chernoff bounds, the number of active streams at time $i$ is $\tO(\beta \cdot L\cdot \floor{p/z})
=\tO(1+\frac 1k\sqrt {\min (\eps^{-1}k,m)}\eps^{-2})=
\tO(1+\frac {\min(\eps^{-2.5}\sqrt k,\eps^{-2}\sqrt m)}{k})$
with high probability.
By a union bound over all the indices, we have that with high probability the number of active streams is $\tO(1+\frac {\min(\eps^{-2.5}\sqrt k,\eps^{-2}\sqrt m)}{k})$ at all times.
For each active stream, the processing of $T[i]$ costs $\tilde O(1)$ time (due to \cref{lem:multistream_dictionary}). 
Moreover, the time cost for updating the data structure of \cref{lem:streams-scheduler} is linear in the number of active streams.
The algorithm spends $\tilde O(1)$ time for each $\ell\in[L]$ to compute $c_{i-m+1}^{(\ell)}$, summing up to a total of $\tilde O(\eps^{-2})$ time.
Therefore, with high probability, the time cost pre character is $\tilde O(\frac {\min(\eps^{-2.5}\sqrt k,\eps^{-2}\sqrt m)}{k}+\eps^{-2})$.

By \cref{lem:generic_algo_correct}, $\hat d_{i-m+1}$ is an $(\eps,k)$-estimation of $d_{i-m+1}$ with large constant probability for each index $i$.
In order to amplify the correctness probability, $\O(\log n)=\tilde O(1)$ instances of the described algorithm are run in parallel, and using the standard median of means technique, the correctness probability becomes $1-n^{-\Omega(1)}$ with just an $O(\log n)$ multiplicative overhead in the complexities.
Hence, \cref{thm:streaming-prob2} follows.

\subsection{More General Problems}

We consider the following generalization of \cref{prob:apx}.


\defproblem{prob:apx-k}{Approximate Text-To-Pattern Hamming Distances with a Fixed Threshold}{A pattern $P\in \Sigma^m$, a text $T\in \Sigma^n$, a distance threshold $k\le m$, and an error parameter $\eps\in(0,\frac13]$.}{For every $i\in [n-m+1]$, a value $\DD_i$  that is an $(\eps,k')$-estimation  of $d_i$ for all $k'\le k$.
}

Notice that \cref{prob:apx} is a special case of \cref{prob:apx-k} with $k=m$.
The solution for \cref{prob:apx-k} is based on the solution for \cref{prob:fixedk}; the reduction is  similar to the reduction described in \cref{sec:preliminaries}. The only difference is that we use only thresholds which are powers of 2 up to $k$ (instead of powers of 2 up to $m$).

An additional speedup is obtained as follows: to cover all values of $k'$ that are smaller than $\eps^{-1}$ we use the exact algorithm of~\cite{CKP19} which takes $\tO(\sqrt {\eps^{-1}})$ time and uses $\tO(\eps^{-1})$ words of space.
Thus, the running time of the algorithm becomes 
\[\tO\left(\sqrt {\eps^{-1}}+\sum_{ \substack{\eps^{-1}\le k'\le k \\  k' \text{ is a power of } 2 }}\frac {\min(\eps^{-2.5}\sqrt {k'},\eps^{-2}\sqrt m)}{k'}+\eps^{-2}\right)\,=\,\tO({\eps^{-0.5}}+\eps^{-2})\,=\,\tO(\eps^{-2}).\]
The space usage of the algorithm is 
\[\tO\left(\eps^{-1}+\sum_{ \substack{\eps^{-1}\le k'\le k \\  k' \text{ is a power of } 2 }}\min\left(\eps^{-2.5}\sqrt {k'},\eps^{-2}\sqrt m\right)\right)\,=\,\tO(\min(\eps^{-2.5}\sqrt k,\eps^{-2}\sqrt m)).\]
The following result follows.

\begin{theorem}\label{thm:streaming-kmm}
	There exists a streaming algorithm for \cref{prob:apx-k} using $\tO(\min(\eps^{-2.5}\sqrt {k},\eps^{-2}\sqrt m))$ words of space and costing $\tO(\eps^{-2})$ time per character.
	For every $i\in [n]\setminus[m-1]$, after the arrival of $T[i]$, the algorithm reports $\tilde d_{i-m+1}$ which with high probability is an $(\eps,k')$-estimation of $d_i$ for all $k'\le k$.
\end{theorem}

In \cref{sec:streamingB}, we introduce another streaming algorithm which uses a different sampling method in order to improve the $\eps$-dependence in space usage of the algorithm at the cost of degraded $\eps$-dependence in the running time of the algorithm.

\bibliographystyle{plainurl}
\bibliography{Refs}

\appendix
\part*{Appendix}

\section{Speed-Ups via Bit Packing}\label{app:lin}

In this appendix, we use fancier bit-packing and table-lookup
tricks in order to further reduce the running time of our offline approximation algorithm in \cref{sec:combo}.
In this setting, we assume that the input strings are \emph{packed}, with each character stored in $\ceil{\log \sigma}$ bits,
so that the input strings take $\Oh(\frac{n \log \sigma }{\log n})$ space only.
Our ultimate goal is to improve the running time for \cref{prob:apx} to $\Oh(\frac{n\log \sigma}{\log n} + \frac{n \log^2 \log n}{\eps^2 \log n})$,
a speedup by essentially a $\Theta(\frac{\log n}{\log \log n})$ factor compared to the running time $\Oh(\eps^{-2} n \log \log n)$
in \cref{cor:combo}.
Note that the $\Theta(\frac{n\log \sigma}{\log n})$ term is necessary because our algorithm in particular locates the exact occurrences of $P$ in $T$,
and this requires reading all the characters of $P$ and $T$ in the worst case.
\subsection{Algorithm of \cref{sec:offline}}

The first part is to speed up the query cost
of the main algorithm in \cref{thm:main} by applying
bit packing to the query array $Q$.
By dividing $[n]$ into blocks of size $b$, we can
encode $Q$ in $\log b$ bits per element, plus one word per block,
for a total of $O(\frac{n}{b} + |Q|\frac{\log b}{\log n})$ words (with word size $w=\Theta(\log n)$).
The output can be encoded in $O(|Q|\frac{\log \log_{1+\eps} m}{\log n})$
words: we report $\lfloor{\log_{1+\eps} \DD_i}\rfloor$ instead of $\DD_i$ for each location $i$, which corresponds to rounding $\DD_i$ down to the nearest power of $(1+\eps)$.\footnote{
	This rounding increases the approximation ratio by a factor $1\pm \eps$. Hence, the algorithm needs to be called with $\teps = \Theta(\eps)$ instead of $\eps$ in order to match the original guarantees. Below, we ignore this technical issue.
}

\begin{theorem}\label{thm:main2}
	For every $s= n^{O(1)}$ and $b$ such that $\log n \le b=n^{o(1)}$,
	\cref{prob:fixedk} can be solved in
	\[O\left(\sqrt{\tfrac{ s mn \log m}{\eps^5 k}}\log s \,+\, \tfrac{n \log s}{\eps^2 k} \,+\,\tfrac{n}{\eps^2 b} \,+\, |Q|\tfrac{\log b}{\eps^2 \log n}\right)\]
	time using a randomized algorithm whose error probability for each fixed $i\in Q$ is $\Oh(1/s)$.
	\end{theorem}

\begin{proof}
We modify the algorithm in \cref{thm:main} starting from the version with the fingerprint functions $F^{(\ell)}$ are combined
with the hash functions $h^{(\ell)}$ resulting in binary values $x_u^{(\ell)}$ and $y_v^{(\ell)}(i)$.

We speed up Lines~\ref{ln:query-generic}--\ref{ln:output} of the algorithm as follows.
Divide $[L]$ into $O(\tfrac{1}{\delta\eps^2})$ intervals of length $L_0=\delta\log n$ for a sufficiently small constant $\delta > 0$.  Fix one such interval $\Lambda$.

Divide $[n]$ into blocks $I$ of size at most $b$
such that the index $v_i$ and the strings $Y_{v_i}^{(\ell)}(i)$ for all $\ell \in \Lambda$ stay unchanged across all indices $i\in I$ in every fixed block.
With probability $1-\Oh(1/s)$, the number of blocks is \[O\left(\tfrac{n}{b}+\tfrac{n}{z}+\beta \tfrac{n p}{z}L_0\right)\,=\,O\left(\tfrac{n}{b}+\sqrt{\tfrac{snm\log m}{\eps k}}\log s + \tfrac{n\log s}{k}\right).\]
We create the following function:
\begin{quote}
\emph{Input}:
the bit vector $\langle y_{v_i}^{(\ell)}(i):\ell\in\Lambda\rangle$
common to indices $i\in I$,
a subset $Q'\sub Q\cap I$ of size at most $\tfrac{\delta\log n}{\log b}$ (with indices stored relative to $\min Q'$),
and the number $u_{\min Q'}$.

\emph{Output}: the counts $|\{\ell\in\Lambda: x_{u_i}^{(\ell)}\neq y_{v_i}^{(\ell)}(i)\}|$ for all $i\in Q'$.
\end{quote}

Note that we do not need the actual index $\min Q'$.  There is enough information in the input to deduce the output,
since $u_i=u_{\min Q'}+i-\min Q'$ holds for all $i\in Q'$.

We bound the input/output size of the function:
The bit vector requires $O(L_0)=O(\delta\log n)$ bits.
The subset $Q'$ requires $O(\tfrac{\delta\log n}{\log b} \cdot \log b) = O(\delta\log n)$ bits.
The number $u_{\min Q'}$ belongs to $[z]$.
The counts require $O(\tfrac{\delta\log n}{\log b}\cdot (\log b + \log L_0)) = O(\delta\log n)$ bits.

For every possible input, the algorithm precomputes the output in $O(bL_0)=O(b\log n)$ time each and stores it in a table.  Since the
number of possible inputs is bounded by $zn^{O(\delta)}$,
the table size and precomputation time are $O(bzn^{O(\delta)})= O(bn^{1/2+O(\delta)})=n^{1-\Omega(1)}=o(\frac{n}{b})$.

For indices $i\in Q$,
the algorithm computes the counts $|\{\ell\in\Lambda: x_{u_i}^{(\ell)}\neq y_{v_i}^{(\ell)}(i)\}|$
using $O(\tfrac{n}{b}+\sqrt{\tfrac{snm\log m}{\eps k}}\log s + \tfrac{n\log s}{k})$ calls to the above function, plus the same number of word operations.

The analysis above is for a fixed interval $\Lambda$. The algorithm performs the same procedure
for all $O(\eps^{-2})$ intervals $\Lambda$, and thus the time bound gets multiplied by an $O(\eps^{-2})$ factor. 
Note that the packed lists of counts can be combined in time linear in the number of words.
\end{proof}

\subsection{Algorithms of \cref{sec:small}}

Next, we speed up the algorithm in \cref{thm:small} by
applying bit packing to the input text string, when the alphabet size $\sigma$ is small.  The text requires $O(n\log\sigma)$ bits and is assumed to be stored in $O(\tfrac{n \log\sigma}{\log n})$ words.

\begin{theorem}\label{thm:small2}
	For every $s= n^{O(1)}$, 
	\cref{prob:fixedk} can be solved in time
	\[O\left(\tfrac{m^2\log s  \,+\, n \log(\sigma\log n)}{\eps^2 \log n}\right),\]
	using a randomized algorithm whose error probability for each fixed $i\in Q$ is $\Oh(1/s)$.
	\end{theorem}	
\begin{proof}
We modify the algorithm in \cref{thm:small} in order to speed up Lines~\ref{ln:query-generic}--\ref{ln:output} without explicitly needing \cref{ln:constructY}. This is achieved as follows. Divide $[L]$ into $O(\tfrac1{\delta\eps^2})$
intervals of length $L_0=\delta\log n$ for a sufficiently small constant $\delta>0$.

Fix one such interval $\Lambda$.
Let $b=\tfrac{\delta\log n}{\log(\sigma\log n)}$. Create the following function:
\begin{quote}
\emph{Input}:
a bit vector $\langle y_{i_0}^{(\ell)}:\ell\in\Lambda\rangle$, the number $i_0\bmod m$, and substrings $T[i_0\dd i_0+b-1]$ and $T[i_0+m\dd i_0+m+b-1]$.

\emph{Output}: the counts $|\{\ell\in\Lambda: x_{i\bmod m}^{(\ell)}\neq y_i^{(\ell)}\}|$ for all $i\in [i_0,i_0+b)$, and
the new bit vector $\langle y_{i_0+b}^{(\ell)}:\ell\in\Lambda\rangle$.
\end{quote}

Note that we are not given the actual index $i_0$; knowing
$i_0\bmod m$ is sufficient.  There is enough information in the input to deduce the output.  From the bit vector
$\langle y_{0}^{(\ell)}(i_0):\ell\in\Lambda\rangle$
and the substrings $T[i_0\dd i_0+b-1]$ and $T[i_0+m\dd i_0+m+b-1]$, we can determine
the bit vectors $\langle y_{0}^{(\ell)}(i):\ell\in\Lambda\rangle$ for all $i\in [i_0,i_0+b]$.

We bound the input and output size of the function:
The bit vectors require $O(\delta\log n)$ bits.
The substrings require $O(b\log\sigma)=O(\delta\log n)$ bits.  The number $(i_0\bmod m)$ belongs to $[m]$.
The output counts require $O(b\log L_0)=O(\delta\log n)$ bits.

For every possible input, the algorithm precomputes the output in $O(bL_0)=O(\log^{O(1)}n)$ time each,
and stores it in a table.
Since the number of different inputs
is bounded by $O(mn^{O(\delta)})$,
the table size and precomputation time is $O(mn^{O(\delta)})$.

For all $i\in [n-m+1]$ (we may as well take $Q=[n-m+1]$),
the algorithm computes the counts $|\{\ell\in\Lambda: x_{u_i}^{(\ell)}\neq y_{v_i}^{(\ell)}(i)\}|$
using $O(\tfrac{n}{b})=O(\tfrac{n \log(\sigma\log n)}{\log n})$ calls
to the above function.

The analysis above is for a fixed interval $\Lambda$. The algorithm performs the same procedure
for all $O(\eps^{-2})$ intervals $\Lambda$, and thus the time bound gets multiplied by an $O(\eps^{-2})$ factor. 
\end{proof}

\begin{theorem}\label{thm:filter2}
	For every constant $\delta > 0$, there is a randomized algorithm for \cref{prob:fixedk} with $k\ge \eps^{-1}m^\delta$ that runs in $\Oh(\frac{n\log \sigma + n\log \log n}{\eps^2 \log n})$ time and is correct with high probability.
\end{theorem}
\begin{proof}
If $m \le \log^{2/\delta} n$, we run the algorithm in \cref{thm:small2} to solve \cref{prob:fixedk}
in time
\[\Oh(\tfrac{n \log(\sigma \log n)}{\eps^2 \log n})=\Oh(\tfrac{n\log \sigma + n\log \log n}{\eps^2 \log n}).\]
Otherwise, we run the algorithm of \cref{thm:main2} with $s=n^{\delta/2}$, $b=\log^2 n$, and $Q=[n]$ to solve \cref{prob:fixedk} in time
\begin{align*}& \Oh\left(b\sqrt{\tfrac{snm \log m}{\eps^5 k}}\log s + \tfrac{n \log s}{\eps^2 k}+ \tfrac{n}{\eps^2 b}+\tfrac{n\log b}{\eps^2\log n}\right)\\
     & =\, \Oh\left(\eps^{-2}\sqrt{n^{1+\delta/2}m^{1-\delta}}\log^{1.5} n + \tfrac{n\log n}{\eps m^\delta}+ \tfrac{n}{\eps^2 \log^2 n} + \tfrac{n \log \log n}{\eps^2 \log n}\right) \\
     & =\,  \Oh\left(\eps^{-2}n^{1-\delta/4+o(1)} + \tfrac{n}{\eps \log n} + \tfrac{n}{\eps^2 \log^2 n} + \tfrac{n \log \log n}{\eps^2 \log n}\right) \,=\, \Oh( \tfrac{n \log \log n}{\eps^2 \log n}).
\end{align*}
The whole algorithm is repeated $\Oh(1)$ times to lower the error probability.
\end{proof}

\subsection{Algorithm of \cref{sec:exact}}

Next, we describe bit-packed variants of \cref{thm:periodic,thm:exact}:

\begin{theorem}\label{thm:periodic2}
	Given an integer $\rho=O(d)$ which is a $d$-period of $P$ and $T$, 
	\cref{prob:exact}, with $\down{\log_{1+\eps} d_i}$ reported instead of $d_i$,
	can be solved in $\Oh(d^2+\frac{n\log \sigma + n\log w \log \log_{1+\eps} k}{w})$ time using a randomized algorithm that returns correct answers with high probability, after $\Oh(k^2 + 2^{O(w)})$-time preprocessing depending only on the threshold $k$ and the machine word size $w$.
\end{theorem}
\begin{proof}
	First, note that the functions $\Delta_{\rho}[T_a]$ and $\Delta_{\rho}[P_a]$ can be constructed in $\Oh(d+\frac{n\log \sigma }{w})$ time
	by iterating through the mismatches between positions at distance $\rho$ since the total support size of $(\Delta_{\rho}[P_a])_{a\in \Sigma}$ and $(\Delta_{\rho}[T_a])_{a\in \Sigma}$ is $\Oh(d)$.
	The algorithm constructs $\Delta^2_\rho[T \otimes P]$ in $\Oh(d^2)$ time using \cref{lem:offline-convolution}.
	Due to \cref{lem:crosscorrelationHam}, these values can be easily transformed into $\Delta^2_\rho[D]$, where $D(i) = d_i$ for $i\in [n-m+1]$
	and $D(i)=0$ otherwise. In particular, $\Delta^2_\rho[D]$ has $\Oh(d^2)$ non-zero entries which can be computed in $\Oh(d^2)$ time.

	The algorithm first determines the values $\down{\log_{1+\eps} d_i}$ for indices $i$ grouped by $i \bmod \rho$. It then interleaves these $\rho$
	subsequences into a single output sequence.
	Observe that if $\Delta^2_\rho[D](i+2\rho)=\cdots =  \Delta^2_\rho[D](i+\ell\rho) = 0$, then $D(i), D(i+\rho),\ldots, D(i+\ell\rho)$ forms an arithmetic progression. Hence, the subsequent values $d_i$ with fixed $r=i\bmod \rho$
	can be decomposed into arithmetic progressions. The total number of these progressions across $r\in [\rho]$ is $\Oh(d^2)$. 
	This bound increases at most twofold if the values are capped with $k+1$, i.e., if $\max(d_i,k+1)$ is represented instead of $d_i$.

	In the preprocessing, the algorithm creates a sequence $S$ of length $O(k^2)$ containing (as contiguous subsequences) all non-constant arithmetic progressions with integer values between $0$ and $k+1$.
	Each entry $s_i$ is then replaced with $\down{\log_{1+\eps} s_i}$, and the latter values is packed in $\Oh(\log \log_{1+\eps }k)$ bits.
	This step takes $\Oh(k^2)$ preprocessing time.
	As a result, for every length-$\ell$ arithmetic sequence in $d_i$, we can copy the sequence of values $\down{\log_{1+\eps} d_i}$ in $\Oh\big(1+\frac{\ell \log  \log_{1+\eps} k}{w}\big)$ time from the appropriate part of the sequence $\down{\log_{1+\eps} s_i}$. The total processing time is $O\big(d^2 +\frac{n\log \log_{1+\eps} k}{w}\big)$.

	Finally, the algorithm needs to interleave the values $\down{\log_{1+\eps} d_i}$ across distinct remainders modulo $\rho$.
	This task is equivalent to transposing an $\ceil{n / \rho}\times \rho$ matrix (with $\Oh(\log \log_{1+\eps} k)$-bit entries) into an $\rho \times \ceil{n/\rho}$ matrix,	and thus it takes $\Oh(\frac{n\log w \log \log_{1+\eps} k}{w})$ time~\cite{DBLP:journals/jal/Thorup02}.
\end{proof}

\begin{theorem}\label{thm:exact2}
	\cref{prob:exact}, with $\down{\log_{1+\eps} d_i}$ reported instead of $d_i$,
	can be solved in $\Oh(\frac{nk^{2/3}}{m^{1/3}}+\frac{nk^2}{m}+\frac{n\log \sigma + \log \log n \log \log_{1+\eps} k}{\log n})$ time using a randomized algorithm that returns correct answers with high probability.
\end{theorem}
\begin{proof}
	We assume that $k \le \sqrt{m}$; otherwise, the $\Oh(n+\tfrac{nk^2}{m})$ running time of \cref{thm:exact} is already good enough.
	First, the algorithm uses \cref{thm:filter} with $\eps = \frac13$ and ${\tilde k} = (m k)^{1/3}$, which results in a sequence $\DD_i$ satisfying the following two properties
	with high probability: if $\DD_i > \frac43{\tilde k}$, then $d_i > {\tilde k}\ge k$; if $\DD_i \le \frac43{\tilde k}$, then $d_i \le 2{\tilde k}$.
	
	Let $C = \{i \in [n-m+1]: \tilde{d}_i \le \frac43 {\tilde k}\}$. 	
	We consider two cases depending on whether $C$ contains two distinct positions at distance $\rho \le \frac12{\tilde k}$ from each other.
	If $C$ does not contain such two positions, then $|C|=\Oh(\frac{n}{{\tilde k}})$, and the algorithm spends $\Oh(k)$ time for each $i\in C$
	to compute $\min(d_i,k+1)$ using $k+1$ \emph{Longest Common Extension} (LCE) queries, which can be answered in $\Oh(1)$ time after $\Oh(\frac{n\log \sigma}{\log n})$-time preprocessing~\cite{DBLP:conf/stoc/KempaK19}. 
	In this case, the overall running time is therefore $\Oh(\frac{nk}{{\tilde k}}+\frac{n\log\sigma}{\log n}) = \Oh(\frac{nk^{2/3}}{m^{1/3}}+\frac{n\log\sigma}{\log n})$.
	
	It remains to consider the case where $C$ contains two distinct positions at distance $\rho\le \frac12{\tilde k}$ from each other.
	In this case, we claim that the running time is $\Oh(\frac{nk^{2/3}}{m^{1/3}}+\frac{n\log \sigma + \log w \log \log_{1+\eps} k}{w})$
	after preprocessing in time $\Oh(2^{\Oh(w)})$, where $w$ is the machine word size. We set $w=\delta \log n$ for a sufficiently small constant $\delta$
	so that the preprocessing time does not exceed $\Oh(\frac{n}{\log n})$.
 
	This claimed running time is proportional to $n$, so we can assume without loss of generality that $n\le \frac32 m$; otherwise, the text $T$ can be decomposed into parts of length at most $\frac32 m$ with overlaps of length $m-1$.  
	We also assume without loss of generality that $\min C = 0$ and $\max C = n-m$;
 	 otherwise, $T$ can be replaced with $T[\min C \dd \max C + m-1]$ and all indices $i$ with $d_i \le k$ are preserved (up to a shift by $\min C$).

	Now, repeating the argument in \cref{thm:exact}, we conclude that $\rho$ is an $\Oh(\tilde k)$-period of both $P$ and $T$.
	Hence, we can use \cref{thm:periodic2}, whose running time is as promised:
	\[\Oh\left({\tilde k}^2+ \tfrac{n\log \sigma + n\log w \log \log_{1+\eps} k}{w}\right)\,=\,\Oh\left(\tfrac{n k^{2/3}}{m^{1/3}}+\tfrac{n \log \sigma + \log w \log \log_{1+\eps} k}{w}\right).\]
	The additional $\Oh(k^2)$ preprocessing time is dominated by the $\Oh({\tilde k}^2)$ term.
	\end{proof}

\subsection{Algorithm of \cref{sec:combo}}

Putting everything together, we get the final algorithm with
slightly sublinear running time when $\sigma$ is small; in particular, the time bound is at least as good as $O(\eps^{-2}n)$ (in fact, $O(n + \tfrac{n\log^2\log n}{\eps^2\log n})$) for any alphabet size~$\sigma$.

\begin{corollary}

There is a randomized algorithm for \cref{prob:apx} that runs in
$O(\tfrac{n\log \sigma}{\log n} + \tfrac{n\log^2 \log n}{\eps^2 \log n})$ time and is correct with high probability.
\end{corollary}
\begin{proof}
Like in the proof of \cref{cor:combo}, we consider three cases.
\begin{itemize}
\item
Case I: $m\le \log^{9} n$.  We run our algorithm in \cref{thm:small2} to solve \cref{prob:fixedk} in time
$O(\tfrac{n\log \sigma + n\log \log n}{\eps^2 \log n})
=O(\tfrac{n\log\log n}{\eps^2 \log n})$. We can solve \cref{prob:apx} by
examining all $k\le m$ that are powers of 2,
in $O(\tfrac{n\log\log n}{\eps^2 \log n} \log m)= O(\tfrac{n\log^2 \log n}{\eps^2 \log n})$ time.
\item
Case II: distances $d_i \le \eps^{-1} m^{1/3}$ and $m > \log^{9}n$.  We use the exact algorithm in \cref{thm:exact2},
which can compute all such distances in time
\begin{multline*}\Oh(\tfrac{nk^{2/3}}{m^{1/3}}+\tfrac{nk^2}{m}+\tfrac{n\log \sigma + n\log \log_{1+\eps} k \log \log n}{\log n})
=\Oh(\tfrac{n}{\eps^{2/3} m^{1/9}}+\tfrac{n}{\eps^2 m^{1/3}}+\tfrac{n\log \sigma}{\log n} + \tfrac{n \log^2 \log n}{\eps \log n})
\\=O(\tfrac{n\log \sigma}{\log n} + \tfrac{n\log^2 \log n}{\eps^2 \log n}).\end{multline*}
\item

Case III:  distances $d_i >  \eps^{-1}m^{1/3}$ and $m > \log^{9}n$.  We run our
algorithm in \cref{thm:main2} with $ s =n^{1/6}$ and $b=\log^2n$
to solve \cref{prob:fixedk}
in time $O(\sqrt{\tfrac{ s nm \log m}{\eps^5 k}}\log s + \tfrac{n\log s}{\eps^2 k} +\tfrac{n}{\eps^2 b} + |Q|\tfrac{\log \log n}{\eps^2 \log n})$.

To solve \cref{prob:apx}, we consider an auxiliary problem where an interval $I$ is additionally given with a guarantee that $d_i\in I$ for every $i\in Q$.
We select $k$ as approximately the geometric mean of the endpoints of $I$ and run the algorithm of \cref{thm:main2},
which lets us split $Q$ into three sublists: one for indices $i$ for which $\DD_i \in [(1-\eps)k, 2(1+\eps)k]$ are guaranteed to be a good approximation of $d_i$,
one for which $d_i<k$ is guaranteed, and one for which $d_i > 2k$ is guaranteed. The latter two sublists are processed recursively with intervals $I\cap [0,k)$
and $I\cap (2k,m]$, respectively.
Finally, we merge the output from the three sublists.
Note that we can split the query list for $Q$ into $3$ sublists in $O(\tfrac{n}{b} + |Q|\tfrac{\log \log n}{\log n})$ time, i.e., in time linear in the number of machine words.
Similarly, we can merge the output query sublists within $\Oh(\frac{n}{b}+|Q|\frac{\log \log_{1+\eps} n}{\log n})=\Oh(\frac{n}{b}+|Q|\frac{\log \log  n}{\eps \log n})$ time. These running times are dominated by the bound from \cref{thm:main2}.

Initially, the interval $I$ is set as $(\eps^{-1}m^{1/3}, m]$ since distances $d_i \le \eps^{-1} m^{1/3}$ have already been computed in the previous case.
In total, the algorithm of \cref{thm:main2} is called once for each power of two $k$, $\eps^{-1}m^{1/3} \le k \le m$, and the total number of query locations across all these instances
is $O(|Q|\log \log n)$ since the depth of the recursion is bounded by $\log \log n$. The total time is therefore
\begin{multline*}O\left(\sum_{\substack{k>\eps^{-1} m^{1/3} \\ k \text{ is a power of } 2}}\sqrt{\tfrac{ s n m \log m}{\eps^5 k}}\log s + \sum_{\substack{k>\eps^{-1} m^{1/3} \\ k \text{ is a power of } 2}} \tfrac{n\log s}{\eps^2 k} + \tfrac{n\log m}{\eps^2 b} + |Q|\tfrac{\log^2 \log n}{\eps^2 \log n}\right) \\ = O\left(\eps^{-2}\sqrt{n^{7/6}m^{2/3}}\log^{1.5}n + \tfrac{n\log n}{\eps m^{1/3}} + \tfrac{n}{\eps^2 \log n} + |Q|\tfrac{\log^2 \log n}{\eps^2 \log n}\right)=\Oh\left(\tfrac{n\log^2 \log n}{\eps^2 \log n}\right).\end{multline*}
We repeat the algorithm $O(1)$ times to lower the error probability. \qedhere
\end{itemize}
\end{proof}

\section{Alternative Streaming Algorithm}
In this section, we present a different sampling method which allows us to  consider fewer offset patterns and offset texts compared to the approach presented in \cref{sec:generic}. The main idea is to pick a random set of offset patterns (each sampled independently with rate $\beta_P$) and a random set of offset texts (each sampled independently with rate $\beta_T$).
The algorithm has access to the sampled offset patterns and the sampled offset texts, so for every location $i$ the algorithm is able to compute the number
of sampled offset patterns aligned against sampled offset texts that yield a mismatch under this alignment.
Notice that for every location $i$, the set of sampled offset patterns that are aligned against sampled offset texts in a fixed alignment forms a random set of offset patterns with sampling rate $\beta_P \cdot \beta_T$.
Hence, the number of such offset patterns that  mismatch the corresponding offset texts yields a good approximation of the value $d'_i$, which is also a good approximation for $d_i$ by \cref{lem:X}.

\begin{algorithm}[H]
	Pick a random prime $p\in [\hat p,2 \hat p)$\tcr*{$\hat{p} = \eps^{-1} s k \log m$}\label{ln:pick:alt}
	$p := \min(p,m)$\;\label{ln:m:alt}
		Pick two random samples $B_P, B_T\subseteq [p]$ with sampling rates $\beta_P$ and $\beta_T$\label{ln:sample:alt}\;
		\lForEach{$b\in B_P$} {
			$\displaystyle X_b\: =\: \bigodot_{j\in [m]:\ j \bmod p\ =\ b} P[j]$\label{ln:constructX:alt}%
		}
		\ForEach{$b \in B_T$}{
			\lForEach{$i \in [n-m+1]$} {
				$\displaystyle Y_v(i) \:=\: \bigodot_{j\in [m]:\ (i+j) \bmod p\ =\ b} T[i+j]$\label{ln:constructY:alt}%
			}
		}
	\ForEach{$i \in Q$}{\label{ln:query}
		$B_i = \{b \in B_P : (b-i) \bmod p \in B_T\}$\;
		Set $c_i=|\{b\in B_i: X_{b}\neq Y_{(b-i)\bmod p}(i)\}|$ and $\DD_i = \frac{c_i}{\beta_P \beta _T}$\;\label{ln:output:alt}
	}
	\caption{Alternative-Algorithm($T,P,Q,k,\eps,s$)}\label{alg:streamingB}
\end{algorithm}

\begin{lemma}\label{lem:alternative}
If $\beta_P \cdot \beta_T > c \cdot \tfrac{\log s}{\eps^2 k}$ for a sufficiently large constant $c$, then for every $i\in Q$
the value $\DD_i$ computed by \cref{alg:streamingB} is an $(\eps,k)$-estimation of $d_i$ with probability $1-\Oh(1/s)$. 
\end{lemma}
\begin{proof} 
	
	Observe that $B_i$ is a subset of $[p]$ with elements sampled independently with probability $\beta_P \beta_T$.
	Moreover, note that $c_i = |B_i \cap M'_i|$, where $M'_i = M_i \bmod p$. 
	Hence, $\E[c_i] = \beta_P\beta_T d'_i$, $\E[\DD_i]=d'_i$, and the symmetric multiplicative Chernoff bound	for every $\teps = \Theta(\eps)$ yields
	\[\Pr[\DD_i \in (1\pm \teps)d'_i] = 1 - \exp(-\Omega(\teps^2 \beta_P \beta_T d'_i))=1-\exp\left(-\Omega\left(\tfrac{c\cdot d'_i \log s}{k}\right)\right)\]
	 We consider three cases.
	\paragraph*{Case 1: $d_i \in [\tfrac12k, 4k]$.}
	By \cref{lem:X}, $(1-\teps)d_i \le |M_i \bmod p| \le d_i$ holds for every $\teps = \Theta(\eps)$ with probability $1-O(1/s)$
	for the prime $p$ picked in \cref{ln:pick:alt}, and obviously this is still true if $p$ is replaced with $m$ in \cref{ln:m:alt}.
	The following argument is conditioned on that event. In other words, we assume that $(1-\teps)d_i \le d'_i \le d_i$.
	In particular, this yields $d'_i = \Theta(k)$, so the Chernoff bound implies
\[\Pr[\DD_i \in (1\pm \teps)d'_i] = 1-\exp\left(-\Omega\left(\tfrac{c\cdot d'_i \log s}{k}\right)\right) \ge 1- \tfrac{1}{s}\]
	provided that the constant $c$ is large enough.
	Hence, $\DD_i = (1\pm \teps)d'_i = (1\pm \Oh(\teps))d_i$ also holds with probability $1-\Oh(1/s)$. This remains true even if we account for the fact that $(1-\teps)d_i \le d'_i \le d_i$ may fail to be satisfied with probability $O(1/s)$.
	Taking $\teps = \Theta(\eps)$ with a sufficiently small constant factor, we conclude that $(1-\eps)d_i \le \DD_i \le (1+\eps)d_i$ holds in this case with probability $1-\Oh(1/s)$. In particular, $d_i < k$ if $\DD_i < (1-\eps)k$ and $d_i > 2k$ if $\DD_i > 2(1+\eps)k$, so $\DD_i$ is an $(\eps,k)$-estimation of $d_i$.

	\paragraph*{Case 2: $d_i < \tfrac12 k$.} In this case, $d'_i \le d_i < \tfrac12k$, so the Chernoff bound implies 
	\[\Pr[\DD_i \le (1+ \teps)k/2] = 1-\exp\left(-\Omega\left(\tfrac{c\cdot k/2 \log s}{k}\right)\right) \ge 1- \tfrac{1}{s}\]
	provided that the constant $c$ is large enough.
	Hence, $\DD_i \le (1+\teps)\tfrac12k$ holds in this case with probability $1-\Oh(1/s)$. In particular, this is true for $\teps = \eps$.
	Since $\eps \le \frac13$, this implies $\DD_i < (1-\eps)k$ and hence $\DD_i$ is an $(\eps,k)$-estimation of $d_i$.

	\paragraph*{Case 3: $d_i >4k$.} \cref{lem:X} applied to any fixed subset of $M_i$ of size $4k$ implies that $d'_i > (1-\teps) 4k$ holds with probability $1-O(1/s)$. The following argument is conditioned on that event. The Chernoff bound therefore yields that
	\[\Pr[\DD_i \ge  (1- \teps) d'_i] = -\exp\left(-\Omega\left(\tfrac{c\cdot d'_i \log s}{k}\right)\right) \ge 1- \tfrac{1}{s}\]
	provided that the constant $c$ is large enough.
	Hence, $\DD_i \ge (1-\teps)d'_i \ge (1-\Oh(\teps))4k$ holds with probability $1-\Oh(1/s)$. This remains true even if we account for the fact that $(1-\teps)4k \le d'_i$ may fail to be satisfied with probability $O(1/s)$. Taking $\teps = \Theta(\eps)$ with a sufficiently small constant factor, we conclude that $(1-\eps)4k \le \DD_i$ holds in this case with probability $1-\Oh(1/s)$.	Since $\eps \le \frac13$, this implies $\DD_i > 2(1+\eps)k$ and hence $\DD_i$ is an $(\eps,k)$-estimation of $d_i$.
\end{proof}

\subsection{Alternative Streaming Algorithm for \cref{prob:fixedk}}\label{sec:streamingB}

Below, we describe how to implement \cref{alg:streamingB} in the streaming model.

\begin{theorem}\label{thm:streaming2}
	There exists a streaming algorithm for \cref{prob:fixedk} where the pattern $P$ can be preprocessed in advance and the text arrives in a stream so that $\DD_{i-m+1}$ is reported as soon as $T[i]$ arrives. 
	The space usage of the algorithm is $\tO\big(\min\big(\frac{\sqrt k}{\eps^2},\tfrac{m}{\eps \sqrt k}\big)\big)$ words, the running time per character is $\tO(\eps^{-3})$,
	and the outputs are correct with high probability.
\end{theorem}

As in \cref{sec:stream}, let $D=\{X_b:b\in B_P\}$ be a dictionary with $|B_P|$ strings.
For each $v\in B_T$, let $Y_v$ be a stream such that at time $i$ (after the arrival of $T[i]$)
 \[Y_v\ = \bigodot_{j\le i: \ j \bmod p \,=\,v } T[j].\] Notice that after the arrival of $T[i]$, we have that $Y_v(i-m+1)$ is a suffix of $Y_v$.

\paragraph{Preprocessing phase.}
During the preprocessing phase, the algorithm chooses a random prime $p\in [\hat p , 2\hat p)$, and picks two random  samples $B_P,B_T\subseteq [p]$ with sampling rate $\beta_P=\beta_T=\sqrt{\frac{s}{\eps^2k}}$ for a large enough constant $s$.
The algorithm applies the preprocessing of the multi-stream dictionary algorithm of \cref{lem:multistream_dictionary} on $D$ so that the patterns from $D$ can be matched against the streams~$Y_v$ for $v\in B_T$.

\paragraph{Processing phase.}
After the arrival of $T[i]$, the algorithm checks if $(i\bmod p)\in B_T$, and if so the algorithm adds $T[i]$ into $Y_{i\bmod p}$.
The multi-stream dictionary matching algorithm of \cref{lem:multistream_dictionary}, identifies the longest suffix of $Y_{i\bmod p}$ which is a pattern from $D$, if such a pattern exists. This way, the algorithm maintains a pointer $\pi_v$ to the longest pattern from $D$ that is a current suffix of $Y_v$.

\paragraph{Evaluating $\tilde d_{i-m+1}$.}
Recall that $B_i = \{b \in B_P : (b-i) \bmod p \in B_T\}$. In the following lemma, we state that the algorithm can compute $B_i$ efficiently after the arrival of $T[i]$, using an auxiliary data structure.

\begin{lemma}\label{lem:Bi_data-structure}
	There exists a data structure that at any time $i$ reports the set		$B_i = \{b \in B_P : (b-i) \bmod p \in B_T\}$. 
	The space usage of the data structure is $O(|B_P|+|B_T|)$ and the time of the $i$th update is $\tO(|B_i|)$.
\end{lemma}
\begin{proof}
	
A position $b\in B_P$ is called \emph{active at time $i$} if $b\in B_i$. Notice that a position $v$ is active at time $i$ if and only if $(b-i)\bmod p \in B_T$.

The data structure stores the elements of $B_T$ in a cyclic linked list,
and maintains one handle for each element $b\in B_P$.
These handles are stored in a hash table that maps future time-points into linked lists of elements' handles. 
The algorithm preserves the invariant that at any time $i$, the handle of any $b\in B_P$ is stored in the linked list of the smallest $j\ge i$ such that $b$ is active at time $j$ (i.e. $(b-j)\bmod p\in B_T$).
Each handle $b$ is maintained in the linked list of time $j$ with a pointer to the element $(b-j)\bmod p$ in the cyclic linked list of $B_T$.

During an update (incrementing $i$ to $i+1$), the algorithm first reports all the active elements in the linked list of time $i$.
Then, in order to keep the data-structure up-to-date and preserve the invariant, the algorithm computes for each $b\in B_i$ the smallest $j>i$ such that $b$  is active also at time $j$. 
This computation is done in constant time per stream by advancing the pointer to the cyclic linked list of $B_T$.
Then, the algorithm inserts the handle of $b$ into the linked list of time $j$.
Finally, the algorithm removes the empty linked list of time $i$ from the hash table to reduce the space usage.
\end{proof}

The algorithm uses \cref{lem:Bi_data-structure} to retrieve $B_i$.
Then, the algorithm estimates $d_{i-m+1}$ in \cref{ln:output:alt} of \cref{alg:streamingB}.
In order to test whether or not $X_b=Y_{(b-i)\bmod p}(i)$, the algorithm checks if $X_b$ is a suffix of the pattern pointed to by $\pi_{(b-i)\bmod p}$.

\paragraph{Complexity analysis.}
Notice that $|D|=\tO(|B_P|)$, the length of each pattern in $D$ is $\Theta(\frac mp)$, and the number of streams $Y_v$ is also $\tO(|B_T|)$.
Consequently, the space usage of the multi-stream dictionary of \cref{lem:multistream_dictionary} is $\tO(|B_P|+|B_T|)$. The space usage of the auxiliary data structure of \cref{lem:Bi_data-structure}  is also $O(|B_P|+|B_T|)$, so the total space usage of the algorithm is $\tO(|B_P|+|B_T|)$.
To bound this quantity, we introduce the following auxiliary fact.

\begin{fact}\label{lem:chernoff}
	If $X \sim B(n,\beta)$ is a binomial random variable, then $\Pr[X \ge 2n\beta + \log s]\le \frac 1{s}$ holds for every $s\ge 1$.
\end{fact}
\begin{proof}
	Markov's inequality yields the claim:
	\[\Pr[X \ge 2n\beta + \log s] = \Pr[2^X \ge  2^{2n\beta}s ] \le \frac{\E[2^X]}{ 2^{2n\beta}s } = \frac{((1-\beta)+2\beta)^n}{2^{2n\beta}s }=\left(\frac{1+\beta}{2^{2\beta}}\right)^n\frac{1}{s} < \frac{1}{s}.\qedhere \]
\end{proof}

By \cref{lem:chernoff}, the total space usage of the algorithm is $\tO(|B_P|+|B_T|)=\tO(p\sqrt{\frac{s}{\eps^2k}})=\tO(\frac p{\eps\sqrt k})=\tO(\frac{\min(\eps^{-2}k,\eps^{-1}m)}{\sqrt k})$ with high probability.

As for the running time,
notice that by \cref{lem:chernoff} with high probability the size $|B_i|$ (for all $i\in [n]$) is $\tO(p\cdot\frac {s}{\eps^2k})=\tO(\frac p{\eps^2k})=\tO(\frac {\min(\eps^{-1}k,m)}{\eps^2k})=\tO(\min (\eps^{-3},\eps^{-2}\frac mk))$.
After the arrival of each character $T[i]$ the algorithm passes $T[i]$ to at most one stream, which costs $\tilde O(1)$ time (by the algorithm of \cref{lem:multistream_dictionary}).
The evaluation of $\DD_{i-m+1}$ is executed by counting the number of mismatches in all the positions of $B_i$. The time per position in $|B_i|$ is $\tO(1)$;
hence, the total time per character is $\tO(|B_i|)=\tO(\min (\eps^{-3},\eps^{-2}\frac mk))$.

Due to \cref{lem:alternative}, the estimation $\hat d_{i-m+1}$ follows the requirements of \cref{prob:fixedk} with constant probability for each index $i$.
In order to amplify the correctness probability, $\O(\log n)=\tilde O(1)$ instances of the described algorithm are run in parallel, and using the standard median of means technique, the correctness probability becomes $1-n^{-\Omega(1)}$ with just an $O(\log n)$ multiplicative overhead in the complexities.
Hence, \cref{thm:streaming2} follows.

\subsection{More General Problems}

A streaming algorithm for \cref{prob:apx-k,prob:apx} is obtainable from the algorithm of \cref{thm:streaming2} as described in \cref{sec:preliminaries}.
The only difference is that we only use thresholds which are powers of two and are smaller than $k$ (instead of smaller than $m$). The running time of the algorithm is 
\[\sum_{\substack{k'\le k\\ k \text{ is a power of } 2}}\tO(\eps^{-3})\:=\:\tO(\eps^{-3}).\]
The space usage of the algorithm is 
\[\tO\left(\sum_{\substack{k'\le k\\ k \text{ is a power of } 2}}\min(\tfrac{\sqrt {k'}}{\eps^2},\tfrac{\sqrt m}{\eps k'})\right)\:=\:\tO(\min(\eps^{-2}\sqrt k,\eps^{-1.5}\sqrt m)).\]

The following result follows.

\begin{theorem}\label{thm:streaming-kmm2}
	There exists a streaming algorithm for \cref{prob:apx-k} using $\tO(\min(\eps^{-2}\sqrt {k},\eps^{-1.5}\sqrt m))$ words of space and costing $\tO(\eps^{-3})$ time per character. 
	For every $i\in [n]\setminus[m-1]$, after the arrival of $T[i]$, the algorithm reports $\tilde d_{i-m+1}$ which  with high probability is an $(\eps,k')$-estimation of $d_i$ for any $k'\le k$.
\end{theorem}

\end{document}